\documentclass[10pt, final, journal, letterpaper, twocolumn]{IEEEtran}
\ifCLASSINFOpdf
\else
\fi

\usepackage[utf8]{inputenc}
\usepackage{amssymb}
\usepackage{bbm}
\usepackage{amsfonts}
\usepackage[dvips]{graphicx}
\usepackage{times}
\usepackage[]{cite}
\usepackage{amsmath}
\usepackage{array}
\usepackage{amssymb}

\usepackage{stfloats}
\usepackage{slashbox}
\usepackage{graphicx}
\usepackage{footnote}
\usepackage{booktabs}
\usepackage{array}
\usepackage{mathtools}
\usepackage{dsfont}
\usepackage{ntheorem}
\usepackage{subeqnarray}
\usepackage{cases}
\usepackage{threeparttable}
\usepackage{color}
\usepackage{nomencl}
\makenomenclature
\usepackage[numbers,sort&compress]{natbib}
\IEEEoverridecommandlockouts

\newtheorem{proposition}{Proposition}
\newtheorem{lemma}{Lemma}
\newtheorem{proof}{Proof}
\newtheorem{corollary}{Corollary}
\usepackage{flafter}
\usepackage{subfigure}
\usepackage{float}
\usepackage{bbold}
\usepackage{supertabular}
\usepackage{caption}
\usepackage{url}
\usepackage{xcolor}
\usepackage[ruled,vlined,lined,ruled,linesnumbered]{algorithm2e}
\newcommand{\bm}[1]{\mbox{\boldmath{$#1$}}}
\DeclareMathOperator*{\argmax}{arg\,max}

\usepackage{subfigure}
\usepackage{stfloats}
\usepackage{etoolbox}
\renewcommand\nomgroup[1]{%
  \item[\bfseries
  \ifstrequal{#1}{P}{Physics Constants}{%
  \ifstrequal{#1}{N}{Number Sets}{%
  \ifstrequal{#1}{O}{Other Symbols}{}}}%
]}
\allowdisplaybreaks
\begin{document}
\title{\huge{Graph Convolutional Neural Networks for\\ Physics-Aware Grid Learning Algorithms}}
\IEEEaftertitletext{\vspace{-2.2\baselineskip}}
\IEEEoverridecommandlockouts
\author{\IEEEauthorblockN{Tong~Wu,~\IEEEmembership{Member,~IEEE}}, \IEEEauthorblockN{Ignacio~Losada~Carre\~no,~\IEEEmembership{Member,~IEEE}},\\
\IEEEauthorblockN{Anna~Scaglione,~\IEEEmembership{Fellow,~IEEE}},
\IEEEauthorblockN{Daniel Arnold,~\IEEEmembership{Member,~IEEE}},\\
%\vspace{-0.5cm}
% \thanks{Tong Wu, Ignacio Losada Carre\~no and Anna Scaglione are with the Department of ECE at  Cornell University, NY, USA (e-mail: \{tw385, il244, as337\}@cornell.edu). Daniel Arnold is with Lawrence Berkeley National Laboratory (e-mail: dbarnold@lbl.gov). }
% \thanks{This research was supported in part by the Director, Cybersecurity, Energy Security, and Emergency Response, Cybersecurity for Energy Delivery Systems program, of the U.S. Department of Energy, under contract DE-AC02-05CH11231. Any opinions, findings, conclusions, or recommendations expressed in this material are those of the authors and do not necessarily reflect those of the sponsors of this work.}
%
}

%}
\maketitle
\newcommand{\norm}[1]{\left\lVert#1\right\rVert}
\newcommand*\abs[1]{\lvert#1\rvert}
%\vspace{-1.8cm}
\begin{abstract}
This paper proposes novel architectures for spatio-temporal graph convolutional and recurrent neural networks whose structure is inspired by the physics of power systems. %for  the three-phase power systems 
The key insight behind our design consists in deriving
the so-called
graph shift operator (GSO), which is the cornerstone of Graph Convolutional Neural Network (GCN) design, from the power flow equations.
We demonstrate the effectiveness of the proposed architectures in two applications: in forecasting the power grid state and in finding a stochastic policy for foresighted voltage control using deep reinforcement learning. Since our design can be adopted in single-phase as well as three-phase unbalanced systems, we test our architecture in both environments. For state forecasting experiments we consider the single phase IEEE 118-bus case systems; for voltage regulation, we illustrate the performance of deep reinforcement learning policy on the unbalanced three-phase IEEE 123-bus feeder system. In both cases the physics based GCN learning algorithms we propose outperform the state of the art. 
\end{abstract}

\vspace{-0.4cm}
\section*{Nomenclature}
\vspace{-0.2cm}
\textit{Abbreviation}
\addcontentsline{toc}{section}{Nomenclature}
\begin{IEEEdescription}[\IEEEsetlabelwidth{$V_1,V_2,V_3$}]\small
\item[GCN] Graph convolutional neural network.
\item[GRN] Graph recurrent neural network.
\item[GSO] Graph shift operator.
\item[GS] Graph signal.
\item[GSP] Graph signal processing.
% \item[STGCN] Spatio-temporal GCN.
% \item[RGCN] Recurrent GCN.
\item[FNN] Fully connected neural network.
\item[CNN] Convolutional neural network.
\item[RNN] Recurrent neural network.
\item[DRL] Deep reinforcement learning.
\item[PSSE] Power system state estimation.
\item[PSSF] Power system state forecasting.
\end{IEEEdescription}

\textit{Sets}
\addcontentsline{toc}{section}{Nomenclature}
\begin{IEEEdescription}[\IEEEsetlabelwidth{$V_1,V_2$}]
\item[$\mathcal{N},\mathcal{E}$] Sets of grid buses $\mathcal{N} = \{1, \cdots, N\}$ and lines.
%\item[$\mathcal{E}$] Set of lines that are the graph edges $\mathcal{E} \subsetneq \mathcal{N} \times \mathcal{N}$.
\item[$\mathcal{N}_s$]  subset of single-phase buses with smart inverters.
\item[$\mathcal{P}_{mn}$]   Phases of line $(m, n) \in \mathcal{E}$.
\item[$\mathcal{P}_{m}$]   Phases of node $n \in \mathcal{N}$.
\end{IEEEdescription}
\textit{Variables}
\addcontentsline{toc}{section}{Nomenclature}
\begin{IEEEdescription}[\IEEEsetlabelwidth{$V_1,V2$}]
%\item[$\bm{v}$] Vector of voltage phasors over all buses, i.e. $\bm{v} = [\bm{v}_1^\top, \cdots, \bm{v}^\top_N]^\top$.

\item[$\bm{v}_n$]  $\bm{v}_n = [v_{n_\phi} | \phi \in \mathcal{P}_n] \in \mathbb{C}^{\abs{\mathcal{P}_n}  \times 1}$ with phase $\angle{{v}}_{n_\phi}$ and magnitude $\abs{v_{n_\phi}}$.
\item[$\bm{v},\bm{i}$] Vectors  of all voltage current injections. 
\item[$\bm{s}$] Vector  of all apparent power injections, $\bm{s} = \bm{p} + \mathfrak{j}\bm{q}$. %$\mathfrak{j} = \sqrt{-1}$.
\end{IEEEdescription}
\textit{Operators}
\addcontentsline{toc}{section}{Nomenclature}
\begin{IEEEdescription}[\IEEEsetlabelwidth{$V_1,V_2, V_3$}]
\item[$\bm A^T, \bm A^H$ ] The transpose and Hermitian of matrix $\bm A$.
\item[$D(\bm A)$]  The vector of the diagonal elements of $\bm A$.
\item[$diag(\bm a)$]  A diagonal matrix with diagonal entries from $\bm a$.
\item[ $\bm A\circ\bm B$] Hadamard product (entry by entry product).
\item[ $(\bm A)^*$]  Conjugate of a complex vector or matrix.
\end{IEEEdescription}
% $\bm A^H$ the superscript $H$ stands for Hermitian, $D(\bm A)$ is the vector of the diagonal elements of $\bm A$, $diag(\bm a)$ is a diagonal matrix whose diagonal elements are the entries of vector $\bm a$, the operator $\bm A\circ\bm B$ is the entry by entry product, also known as Hadamard product.

%\textit{Abbreviation}
%\addcontentsline{toc}{section}{Nomenclature}
%\begin{IEEEdescription}[\IEEEsetlabelwidth{$V_1,V_2,V_3, V_4$}]
%\item[DRL] Deep Reinforcement Learning.
%\end{IEEEdescription}
% \vspace{-0.5cm}
\section{Introduction}
\subsection{Background and Motivation}
% Voltage Control Background and Challenge
%In distribution systems, voltage profiles are the most critical indicator of the system operating condition, whilst reliable and efficient energy management is the core task \cite{liu2020two, gao2021consensus, sun2021two, zhang2020deep}.
%This is why Volt-VAR control (VVC) schemes have been developed and integrated into  distribution systems to reduce network losses \cite{gao2021consensus}, avoid voltage violations \cite{cao2021deep} and mitigate cyber attacks \cite{roberts2020deep}. However, the rapid growth of distributed energy resources makes it increasingly difficult to manage  voltage profiles on active distribution networks.
%The sensing infrastructure is expected to provide complete and reliable measurements of power systems  for monitoring, control and protections of power grids \cite{}. 
%The three-phase power grid is a large-scale  physical network, where the nodes of the associate graph are the grid buses and its edges are its  distribution lines \cite{carreno2022log}. It is therefore natural to see voltage phasor measurements as graph signals that are indexed by the buses of the distribution line grid \cite{ramakrishna2021gridgraph}. 

The access to high-quality grid sensors data, and phasor measurement units (PMUs) in particular, has prompted a lot of interest in applying advanced learning algorithms to address grid inference and control problems. The main advantage of learning techniques, when compared to regression or optimization problems that purely rely on the physics, lies in their ability to internalize statistical patterns in the training data that are not captured by the physical constraints only.  
On the other hand, while many {\it black-box} learning approaches respond well to the challenge, accounting for the physical equations explicitly, rather than learning them as a pattern, reduces the number of parameters in the model, mitigating  overfitting problems. This is why not all Neural Networks (NN) have the same architecture and, in particular, both time-series and images are best processed by Convolutional or Recursive NN, leveraging the shift or state invariance of the data to reduce the parameter space and increase generalization ability.

The framework of Graph Signal Processing (GSP) has emerged as the the most promising approach to generalize these architectures to data that have the irregular support of a network. GSP provides a natural representation  for both the data (node attributes) and the underlying structure (edge attributes) \cite{dong2020graph}.
Its application to grid measurements has recently spurred significant interest \cite{ramakrishna2021gridgraph}. %The most popular graph filter for the low-pass GS is the polynomial filter, namely Chebyshev filters that generalize the ones for GS \cite{defferrard2016convolutional}. 
However, the linear models of Graph filters have limited capability to learn the possibly complex mappings that are needed for classification, forecasting and for the approximation of optimum control policies when compared to neural network models. In fact,  on several applications that include high-dimensional data from a network structure, GCNs  have shown to have the best generalization capabilities \cite{jiang2021emphatic}.
Motivated by the promise of AI applications to  power system data, the overarching goal of this paper is to develop state of the art GCN architectures for inference and decision making that best capture the spatial and temporal features of grid data that derive from the AC power-flow constraints. 
\vspace{-0.3cm}
\subsection{Related Works}
%The literature on voltage control is vast.  We limit our review to the most relevant literature pertaining DRL and GCN.
%
We first highlight the related prior research on grid GSP and GCN and then the literature on the two applications we consider to showcase the performance benefits of the proposed schemes, namely: power system state estimation and forecasting, and reinforcement learning for voltage control.

\subsubsection{Graph Convolutional Neural Network}
A handful of papers have so far  successfully applied GCN to distribution systems' management, considering applications that include fault localization \cite{chen2019fault}, distribution  system  state  estimation \cite{zamzam2020physics}, and synthetic data generation \cite{liang2020feedergan}.
% Introduce GCN, RGCN, CGCN
 GCNs are a generalization of CNN, aimed at capturing the impact that the network connectivity has on the patterns of the data associated to the network nodes. The foundation of GCN lies in the  GSP definition of graph filters and of what is referred to as the Graph Shift Operator (GSO).  In \cite{chen2019fault}, the GSO is defined as the weighted adjacent matrix, constructed based on the physical distance between nodes.  In \cite{zamzam2020physics}  the authors prune the weights of a conventional Fully connected Neural Network (FNN) based on the power grid topology, without considering the grid lines admittances.
 %, and is regarded as a simplified version of  GCN with the first-order Chebyshev polynomial coefficient. 
In the paper \cite{liang2020feedergan} the GSO is constructed as adjacent matrix that captures the correlation among historical data. All these approaches are not directly considering the electrical characteristics of the overhead power lines.
In prior work \cite{ramakrishna2021gridgraph} we provided ample evidence that the right framework to apply GSP for grid signals should be rooted in the basic network analysis that has been used to model power systems signals for decades \cite{ribeiro2013power}. In fact, Ohm's law is the obvious driver of the correlation in the state vector of the grid, and the GSO can be derived from first principles from the grid physics, which suggests that the right GSO is the admittance matrix itself. While \cite{ramakrishna2021gridgraph} has shown the benefits of using complex graph filtering methods to address a number of inference problems, physics-based Graph Neural Networks architectures are still missing.  We note that Pytorch and Tensorflow \cite{NEURIPS2019_9015, tensorflow2015} operate on real valued nodal data, which are not compatible with the GSO derivation in \cite{ramakrishna2021gridgraph}.  

This paper addresses two gaps left by the prior art. First, we derive a physics inspired GSO  based on the power flow equations in the real domain. Using that we defined two efficient spatio-temporal graph neural network architectures: Graph Convolutional Neural Networks (GCN) and Graph Recurrent Neural networks (GRN). Second, we extend these architectures to the unbalanced three-phase power systems, unleashing the power of physics inspired AI methods to distribution systems. 

\subsubsection{Power System State Estimation and Forecasting} There is a vast literature on Power System State Estimation (PSSE)
%PSSE aims at retrieving the unknown system state, i.e.,  voltage phasors at all buses, given the grid parameters and a set of measurements provided by the supervisory control and data acquisition (SCADA) system 
\cite{zhang2019real}. The use of graph neural networks is very recent \cite{kundacina2022state, hossain2021state}. Neither papers derived the GSO from the power flow equations (e.g. adjacency matrix used in \cite{kundacina2022state}). Also,  \cite{kundacina2022state, hossain2021state} ignored the temporal dependencies of the voltage phasors and focused on state reconstruction, not forecasting. 
Power Systems State Forecasing (PSSF) has so far been pursued via a single-hidden-layer NN in \cite{do2009forecasting1, do2009forecasting2}. Because the number of FNN parameters grows linearly with the length of the input sequences, the proposed methods are prone to overfitting. 
In  Fig. \ref{estimation_forecast} of Section V, we show that our method attains about an order of magnitude improvement in accuracy over the FNN in \cite{do2009forecasting1, do2009forecasting2}.

\subsubsection{Reinforcement Learning for Voltage Control}
Voltage control problems can be modeled as mixed-integer nonlinear  programs that include  an optimal power flow; they are \textit{nonconvex} and \textit{NP-hard} \cite{zhang2020deep} and, therefore, impractical for a real-time implementation.  In recent years several authors have explored Deep Reinforcement Learning (DRL) methods as alternatives to brute force optimization, to search via training approximately optimal policy functions, parametrized as a neural network. 
Existing DRL methods  for Volt-VAR control in distribution grids are broadly classified as value-based \cite{vlachogiannis2004reinforcement, xu2019optimal, yang2019two, duan2019deep} and policy-based RL algorithms \cite{wang2020data, haarnoja2018soft, cao2020multi}. Unfortunately, DRL algorithms can become unstable when combining function approximation, off-policy learning, and bootstrapping (a combination referred to as the deadly triad \cite{van2018deep}). 
Many authors resorted to FNN or CNN architectures which, as we discussed previously, are over-parametrized in their feature extraction layers and, therefore, likely to trigger the {\it deadly triad} of DRL,  i.e., ending up in a lot of instabilities, or no convergence  
\cite{zhang2020deep}. 
%
 
%Finally, to the best of our knowledge, these policies use the entire network state, and the issue of how to use DRL with limited observations and using AMI measurements instead of PMUs, which are more common \cite{hu2019utilizing}, has not been dealt with \cite{zhang2020deep}. There are, however,  
%
%However, most existing DRL methods do not consider AMI measurements  without  voltage phasors, where AMI measurements are likely to trigger the the deadly triad of DRL \cite{cao2021deep}.
%
%
% Last but not least, the current DRL methods for the distribution network control are hard to regulate multiple  devices in a distributed fashion. This is because each agent only has a local observation, and thus its local action  cannot  coordinate with actions from other agents. As a result, they only render  local optimal actions.
%
Similar conclusions apply to the adversarial DRL approach for Volt-VAR control proposed in \cite{liu2020two} for distribution grids. Very recently, \cite{zhao2021learning, lee2021graph} considered the graph correlation of voltage phasors in their DRL design, but ignored the temporal correlation of their time series.
%are omitted and the data compression ability of GCN has not been employed. 
Note that \cite{zhao2021learning, lee2021graph} require the full system state. 
In fact, all aforementioned algorithms require access to measurements of the full state of the system as an input \cite{liu2020two, gao2021consensus}. Even when the state is observable, it is hard to scale these methods to work with large-scale networks with high-dimensional features \cite{sun2021two, zhang2020deep}. This motivates the derivation of a reduced GSO that contracts the network to the buses that are controlled.  

%, and the authors do not provide a clear indication on how that would be made available as an input to the policy. 
%
%In particular, we will start with two basic GCN frameworks, e.g., graph recurrent neural networks and spatio-temporal Graph ConvNet for Volt-VAR control in distribution networks (see Fig. \ref{fig:GCN}.)
%
% The state of art on the third problem to utilize AMI measurements for voltage control adopts either the SDP relaxation that maps AMI measurements to voltage phasors \cite{weng2013distributed} or  the linear relaxation for an approximate solution \cite{bagheri2016assessing}.  These methods have very high computational complexity, and are hard to scale to a large distribution system in real-time control.

% On the one hand,  this kind of methods just input the local observation to each agent, where each agent could only provide a local optimal action for the global control problems.  On the other hand, full observations requires the global communications of each agent, which is not realistic. 

% fault location, GCN drawback
%power flow calculation, GCN drawback
% data generation, GCN drawback
% What will we do?
\vspace{-0.3cm}
\subsection{Contributions and Organization}
%To address the challenges mentioned above, in this paper we propose a STGCN framework, in two architectures:  GCN GCNs  and CGCNs. The proposed STGCN framework is based on the three-phase power flow equations to design the physical-aware GSO. We further implement the proposed method for power system state estimation and forecasting, and reinforcement learning for voltage control.
Our main contributions are summarized next:
\begin{itemize}
	\item We develop novel physics-aware  Graph Convolutional Neural networks  and Graph Recusive Neural networks   architectures that are applicable to single and  three-phase unbalanced power systems. In each layer the parameters of Chebyshev graph-temporal filters capture the spatio-temporal features of the grid signals. Architecturally, the main novelty is in deriving a real-valued GSO from the physics of power flow equations.
	
    %\item Since the state-of-the-art GCNs currently take only real-valued inputs, we derive a suitable graph shift operator defined as a function of the designed system susceptance matrix that is modified as a Laplacian matrix for the amplitude and phase of the voltage phasor based on the  three-phasor power flow equations. 
%    \item The detailed GSO formula and derivation are shown in \textbf{Appendix} of the full version \cite{wu2022deep}.
    
 	\item To deal with sparse deployments of PMUs, or simply for scalability we show that one can use the  kron-reduced network GSO instead of the full GSO. 
 	\item We demonstrate in two case studies how the proposed architectures outperform the state of the art in PSSE and PSSF, as well as in Volt-Var control.

\end{itemize}

% \subsection{Organization and Overview}
The rest of the paper is organized as follows. In Section II, we briefly review the key notions of GSP, setting the stage in Section III where we derive the physics inspired GSO and introduce our graph neural networks architectures. 
%in three-phase distribution networks. 
In Section IV, we describe two applications of the proposed GCN and GRN frameworks that are tested numerically in Section V. Finally, we conclude the paper in Section VI.

\section{A Brief Review of Graph Signal Processing}
To make the paper self-contained,  we first review the basic theory of  GSP  (more details can be found e.g. in \cite{ramakrishna2021gridgraph})
%\subsubsection{Graph Signal Processing}
%
for a general graph $\mathcal{G} = (\mathcal{V}, \mathcal{L})$, with vertex set $\mathcal{V}$ and edge set $\mathcal{L}$. The concepts defined here will be applied to three-phase distribution network whose graph topology is $\mathcal{D} = (\mathcal{N}, \mathcal{E})$, where the node $n_{\phi}$ of bus $n$ on phase $\phi$ in $\mathcal{D}$ corresponds to the $i^{th} \in \mathcal{V}$ node of $\mathcal{G}$ and the edges are the transmission lines connecting the buses. %In the following, we use $n_{\phi}$ to represent the index of buses on phase $\phi$.  
A graph signal $\bm{x}  \in \mathbb{R}^{\abs{\mathcal{V}}}$ (which, in the grid, is mainly the state vector) is a vector indexed by the network nodes.
The set $ \mathcal{N}_i$ denotes the subset of nodes connected to node $i$, i.e. node $i$'s neighborhood.  A GSO is a matrix $\mathbf{S}  \in \mathbb{R}^{\abs{\mathcal{V}} \times \abs{\mathcal{V}}}$ that linearly combines graph signal neighbors' values. Almost all operations including filtering, transformation and prediction are directly related to the GSO \cite{ramakrishna2021gridgraph} which generalizes the  $s$ variable representing the derivative in the Laplace domain for signals in time. Consistent with the intuition that it should operate as a differential operator, the GSO, denoted by  $\mathbf{S}\in \mathbb{R}^{\abs{\mathcal{V}} \times \abs{\mathcal{V}}}$, is usually chosen as a graph weighted Laplacian:
\begin{equation}
\begin{split}\small
[\mathbf{S}]_{ij} =  \begin{cases}
\sum_{k\in{\mathcal{N}_i}} S_{i,k}, & i = j, \\
-S_{i,j}, & i \neq  j.
\end{cases}
\end{split}
\end{equation}
In this work, we focus on real symmetric GSOs, i.e., $\mathbf{S} = \mathbf{S}^\top$ that are appropriate for power grid applications. 
A graph filter is a linear matrix operator $\mathcal{H}(\mathbf{S})$, function of the GSO, that operates on graph signals as follows  
\begin{align}\label{gs}
    \bm{w} = \mathcal{H}(\mathbf{S}) \bm{x}.
\end{align}
What defines the dependency of $\mathcal{H}(\mathbf{S})$ on the GSO is that is that $\mathcal{H}(\mathbf{S})$ must be shift-invariant (like a linear time invariant filter in the time domain), i.e. $\mathcal{H}(\mathbf{S})\mathbf{S}=\mathbf{S}\mathcal{H}(\mathbf{S})$. This is possible if and only if $\mathcal{H}(\mathbf{S})$ is a matrix polynomial
\footnote{Note that the graph filter order $K$  can be infinite.}:
\begin{equation}
\begin{split}\label{lsi}
 \mathcal{H}(\mathbf{S}) = \sum_{k=0}^{K} h_k  \mathbf{S}^k.
\end{split}
\vspace{-0.2cm}
\end{equation}
Let the eigenvalue decomposition be $\mathbf{S} = \mathbf{U}\bm{\Lambda} \mathbf{U}^\top$, where $\bm{\Lambda}$ is a diagonal matrix with eigenvalues $\lambda_1\leq  \dots\leq\lambda_{\abs{\mathcal{V}}}$. Since the GSO $\mathbf{S}$ is symmetric
$\mathbf{U}$ is unitary and is the basis for Graph Fourier Transform (GFT). The GFT of a graph signal is, therefore, $\tilde{\bm x}=\mathbf{U}^\top\bm x$ and the eigenvalues $\lambda_{\ell}, \ell=1,\ldots, \abs{\mathcal{V}}$ are the \textit{graph frequencies}. 
%A linear shift-invariant operator $\mathcal{H}$ that is defined as:
%\begin{equation}
%\begin{split}\label{lsi}
%\bm{w} = \mathcal{H}(\mathbf{S}) \bm{x}, %\mathcal{H}(\mathbf{S}) = \sum_{k=0}^{K} h_k  \mathbf{S}^k,
%\end{split}
%\end{equation}
%where linear shift-invariant filters should satisfy the condition: $\mathbf{S} \mathcal{H}(\mathbf{S}) = \mathcal{H}(\mathbf{S}) \mathbf{S}$. 
From \eqref{lsi} it follows that:
\begin{equation}
\begin{split}
\mathcal{H}(\mathbf{S}) = \mathbf{U} \bigg(\sum_{k=0}^{K} h_{k} \Lambda^k \bigg) \mathbf{U}^{-1}.
\end{split}
\vspace{-0.2cm}
\end{equation}
The matrix $\sum_{k=0}^{K} h_{k} \Lambda^k$  is a diagonal, with $i^{th}$ entry $\tilde{h}(\lambda_i) \triangleq \sum_{k=0}^{K} h_k \lambda^k_i $.  Hence, $\tilde{\bm h}=[\tilde{h}(\lambda_1),\ldots,\tilde{h}(\lambda_{|\cal V|})]$ is the transfer function for graph filters. In the GFT domain this yields:
\begin{equation}
\begin{split}
\bm{w} = \mathcal{H}(\mathbf{S})\bm{x}  \iff \tilde{\bm{w}} = \tilde{\bm{h}} \circ \tilde{\bm{x}}.
\end{split}
\end{equation}
For time series of graph signals $\{\bm{x}_t\}_{t\ge 0}$ one can use graph temporal filters models:
\begin{equation}
\bm{w}_t  = \sum_{\tau = 0}^t \mathcal{H}_{t - \tau}( \mathbf{S}  ) \bm{x}_\tau~~~~\mathcal{H}_t(\mathbf{S}) = \sum_{k=0}^K h_{k,t} \mathbf{S}^k,
\end{equation}
and harness DSP tools, defining a combined GFT and $z-$transform for their analysis:
\begin{equation}
\begin{split}
\mathbf{X}(z) = \sum_{t=0}^{T-1} \bm{x}_t z^{-t}, ~~~\tilde{\mathbf{X}}(z) = \mathbf{U}^\top\mathbf{X}(z).
\end{split}
\end{equation}
where $T$ is the length of the graph signal time series. 
In particular, for filter of order $T$,  $\mathbf{S} \otimes z$ is the graph temporal GSO and a graph Spatio-Temporal filter is defined as follows:
\begin{eqnarray}
%\mathcal{H}_t(\mathbf{S}) = \sum_{k=0}^K h_{k,t} \mathbf{S}^k  \iff
\mathcal{H}(\mathbf{S} \otimes z) = \sum_{k=0}^K H_k(z) \mathbf{S}^k,~~
 H_k(z)=\sum_{t=0}^{T-1} h_{k,t}z^{-t}
\end{eqnarray}
i.e. $H_k(z) $ is the $z-$transform of the filter coefficients $h_{k, t}$.
 In the $z$-domain the input output relationship is:
\begin{equation}
\mathbf{W}(z) = \mathcal{H}(\mathbf{S} \otimes z) \mathbf{X}(z),
\end{equation}
 The graph-temporal joint transfer function is: 
%in the GF domain is expressed as:
\begin{equation}
\begin{split}\label{stkernal}
\mathbb{H}(\bm \Lambda, z) = \sum_{k=0}^K \sum_{t=0}^{T-1} h_{k, t} \bm \Lambda^k z^{-t}.
\end{split}
\end{equation}
which is a diagonal matrix.
Denoting by $\tilde{\mathbf{X}}(z)=\mathbf{U}^\top \mathbf{X}(z)$, the input-output relationship in the combined GFT-$z$-domain:
\begin{equation}
 \tilde{\mathbf{W}}(z) =\mathbb{H}(\bm \Lambda, z) \tilde{\mathbf{X}}(z).
\end{equation}
%which is again an element by element multiplication since $\mathbb{H}(\bm \Lambda, z)$ is a diagonal matrix.

\section{Graph Convolutional Neural Networks Architectures for Power Grid Signals}
%In this section, we first review both GSP  and GCN. Then, we derive a physics-aware GSO based on the three-phase power flow equations and conclude the section with the description of the proposed ST-GCN and ST-GRN architectures. 
%\vspace{-0.2cm}

In the left of  Fig. \ref{fig:CGCN_RGCN}, we  illustrate  differences of GCN, CNN and FNN applied to the same graph signal with specific neuron structures for each time instant. Observe that GCNs are  generalizations of CNNs where time-series filters are replaced by application-dependent graph temporal graph filters.
In fact, in a ST-GCN, the weights $h_{k, t}$ are the parameters learnt during training in the feature extraction layers  \cite{defferrard2016convolutional}. 
In FNNs, instead, one may train matrices with arbitrary weights, and therefore the parameters grow in the order of the number of nodes and filter memory squared. Next, we will introduce the physics-based derivation of the GSO for the grid GCN.
\begin{figure*}[!htb]
	\centering
	\includegraphics[width=2.00\columnwidth]{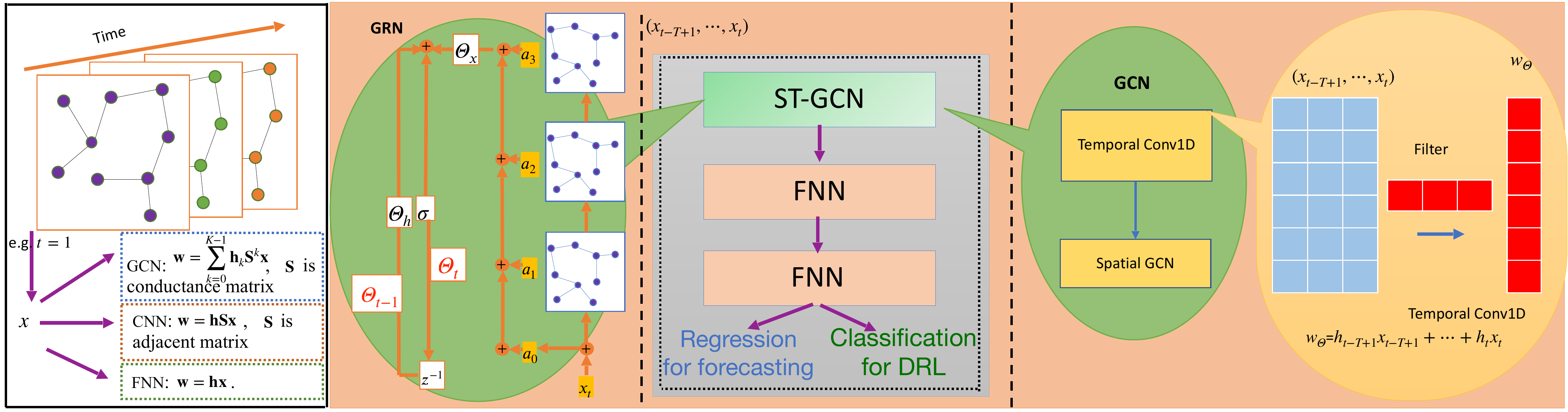}
	\caption{(Left) Information flow of the power systems, (middle) the GRN  structure achieved by the RNN and  GCN blocks  and (Right)  the  GCN structure.}% for the different intervals as the algorithm iterates.\textcolor{red}{Kari, can you please clarify what the price is for?}}
	\label{fig:CGCN_RGCN}
	\vspace{-0.7cm}
\end{figure*}
%
% \begin{figure}[!htb]
% 	\centering
% 	\includegraphics[width=0.9\columnwidth]{singleagent_DRL.pdf}
% % 	\includegraphics[width=0.92\columnwidth]{multiagent_DRL.pdf}
% 	\caption{Spatio-Temporal Graph Convolutional (STGCN) Structure achieved by the Conv1D and Spatial GCN blocks.}% for the different intervals as the algorithm iterates.\textcolor{red}{Kari, can you please clarify what the price is for?}}
% 	\label{fig:agent_drl}
% 		\vspace{-0.3cm}
% \end{figure}

%\subsubsection{Grid Graph Signal Processing}
\vspace{-0.3cm}
\subsection{Real-Valued Grid Graph System Operator}
A physics-inspired framework for Grid-GSP was proposed in \cite{ramakrishna2021gridgraph} to provide an interpretation for the spatio-temporal properties of voltage phasor measurements by utilizing the admittance matrix as graph filters.  GCNs tools currently takes only real-valued inputs \cite{ruiz2021graph}.  Next we show how to derive a real-valued, physics inspired, GSO from the power flow equations that can be used to capture the features of 
%. This is the foundation of our design for the proposed physics inspired GCNs and GRNs, taking as an input both amplitude and phase of the state vectors. Our construction fits the intuition that neighboring voltages and phases patterns are the most correlated. 
the real graph signal components represented by the pairs of voltage magnitude $\abs{v_{n_{\phi}}}$ and voltage phase $\varphi_{n_{\phi}}$ vectors. Let $\bm{s} = \bm{p} + \mathfrak{j} \bm{q}$ be the vector of net apparent power at buses ($\bm{s} = [\bm{s}_1^\top, \cdots, \bm{s}_{\abs{\mathcal{N}}}^\top]^\top $),  with the $n^{th}$ entry  $\bm{s}_n = \bm{p}_n + \mathfrak{j} \bm{q}_n, \bm{s}_n \in \mathbb{C}^{ \abs{\mathcal{P}_n}\times 1}$. Further, let $\bm{v}$ and $\abs{\bm{v}}$ be the vectors of bus voltage phasors and magnitudes, respectively, with $\bm{v} \in \mathbb{C}^{ {\sum_{n\in \mathcal{N}} \abs{\mathcal{P}_n}}\times 1}$ and $\abs{\bm{v}} \in \mathbb{R}_+^{^{ {\sum_{n\in \mathcal{N}} \abs{\mathcal{P}_n}}\times 1}}$, and let $\bm{i} \in \mathbb{C}^{^{ {\sum_{n\in \mathcal{N}} \abs{\mathcal{P}_n}}\times 1}}  ~\textnormal{and}~\abs{\bm{i}}\in \mathbb{R}_+^{^{ {\sum_{n\in \mathcal{N}} \abs{\mathcal{P}_n}}\times 1}}$ be the vectors of net bus current phasors and magnitudes, respectively: 
\begin{eqnarray}
    &v_{n_\phi} = \abs{v_{n_\phi}}~e^{\mathfrak{j} \angle{v_{n_\phi}}},~
    i_{n_\phi} =  \abs{i_{n_\phi}}~e^{\mathfrak{j} \angle{i_{n_\phi}}},~ \forall {n} \in \mathcal{N}, {\phi} \in \mathcal{P}_n,\notag\\
 &\bm{i} = \mathbf{Y} \bm{v},  \label{eq:node_injection}
\end{eqnarray}
%\vspace{-0.1cm}
where $\mathbf{Y}$ in \eqref{eq:node_injection} is a block matrix of dimensions $\sum_{n\in \mathcal{N}} \abs{\mathcal{P}_n} \times \sum_{n\in \mathcal{N}} \abs{\mathcal{P}_n}$ and $\mathbf{B}$ is the susceptance matrix. More specifically, the blocks in $\mathbf{Y}$ are:
\begin{enumerate}
    \item the matrices $ \mathbf{Y}_{mn}$, occupying the $\abs{\mathcal{P}_{mn}} \times \abs{\mathcal{P}_{mn}}$ off-diagonal block corresponding to line $(m, n) \in \mathcal{E}$; and,
    \item the $\abs{\mathcal{P}_{mn}} \times \abs{\mathcal{P}_{mn}}$ diagonal block corresponding to node $n\in \mathcal{N}$ with $\mathcal{N}_m = \{n|(m,n)\in\mathcal{E}\}$:
 \begin{align}
 [\mathbf{Y}]_{\mathcal{P}_n, \mathcal{P}_n} = \sum_{m\in \mathcal{N}_n} \left( \frac{1}{2} \mathbf{Y}^{s}_{mn} + \mathbf{Y}_{mn}  \right)
 \end{align}
\end{enumerate}
\vspace{-0.1cm}
% \begin{align}
% [\mathbf{Y}]_{\mathcal{P}_n,\mathcal{P}_n}=
% \begin{cases}
% y^{sh}_{ii} + \sum_{k \in \mathcal{B}_i} y_{ik}  , i = j  \\
% - \bm{y}_{mn}; m \neq n
% \end{cases}
% \end{align}
%with $y_{m_{\phi}n_{\phi}} = g_{m_{\phi}n_{\phi}}+\mathfrak{j} b_{m_{\phi}n_{\phi}}$ as the admittance of the branch between node ${m_{\phi}}$ and ${n_{\phi}}$ when $(m,n) \in \mathcal{E}$.
%
$\mathbf{Y}^{s}_{mn}$ is the shunt element. 
%The GSO is associated with $\bm Y$, which is a complex Laplacian. 
As first noted in \cite{ramakrishna2019modeling} Ohm's law allows us to view voltage as the output {\it low-pass} filter by ${\bm v} ={\mathbf{Y}}^{-1} \bm i$ (an integrator),  implying that $\mathbf{Y}$ is an appropriate GSO.
%Since the state-of-the-art GCNs currently take only real-valued inputs, we will derive a suitable graph shift operator defined as a function of the designed system Laplacian matrix next.
The first step to obtain the GSO in the real domain,  is to express the three-phase power flow equations as follows: 
\begin{align}
&\bm{i}_{n} =  \sum_{m\in \mathcal{N}_n}\left[\Big(\frac{1}{2} \mathbf{Y}^{s}_{mn} + \mathbf{Y}_{mn}^{(n)} \Big) \bm{v}_n +  \mathbf{Y}_{mn}^{(m)}\bm{v}_m\right]\notag
\end{align}
Note that, in general, for distribution lines $\mathbf{Y}_{mn}^{(m)} = - \mathbf{Y}_{mn}^{(n)} $ with the exception of transformer or regulators.  In the following analysis, we omit the influence of transformer or regulators and assume $\mathbf{Y}_{mn}^{(m)} = - \mathbf{Y}_{mn}^{(n)} $. 
The power flowing from bus $n\in\mathcal{N}$ to bus $m\in\mathcal{N}$ is:
\begin{align}
&\bm{s}_{n} = {D}\left(\bm{v}_n \bm{i}_{n}^H\right),~~~~\forall \phi \in \mathcal{P}_n, n\in \mathcal{N}.
\end{align}
Assuming that the susceptance $\mathbf{B}$ dominates over the conductance, denoting the imaginary parts of  the matrices $\mathbf{Y}^{s}_{mn}$, $\mathbf{Y}_{mn}^{(n)}$ and $\mathbf{Y}_{mn}^{(m)}$ respectively by $\mathbf{B}^{s}_{mn}$, $\mathbf{B}_{mn}^{(n)}$ and $\mathbf{B}_{mn}^{(m)}$ (which are symmetric), we have:
\begin{align}\label{sinj}
\bm{s}_{n} \!\approx \!\! 
%&\notag\\&
\sum_{m\in \mathcal{N}_n}\!\!\!-\mathfrak{j}
D\bigg(\bm{v}_n\bm{v}_{n}^H
\Big(\frac{1}{2} \mathbf{B}^{s}_{mn} \!\!+ \!\mathbf{B}_{mn}^{(n)}  \Big)\!\!+ 
\bm{v}_n\bm{v}_{m}^H \mathbf{B}_{mn}^{(m)}\bigg).
\end{align}
From \eqref{sinj} we obtain an approximation of the power flow equations that describes the dependence between the active and reactive power on  $\abs{\bm{v}_n}$ and $\angle{{\bm{v}}_n}$  by approximating quadratic terms $\bm{v}_n\bm{v}_{n}^H$  and $\bm{v}_n\bm{v}_{m}^H$  as follows:

\begin{lemma}
With a first order expansion of the phase term the quadratic term $\bm{v}_n\bm{v}_{m}^H$ can be approximated as:
\begin{align}
&\bm{v}_n\bm{v}_{m}^H \approx \!\!\!
&\left(\mathbb{1}\mathbb{1}^\top\!  +\!\mathfrak{j}(\bm{\varphi}_n\mathbb{1}^\top \!\!\!-\! \mathbb{1}\bm{\varphi}_m^\top )\right) \circ \Big(\text{diag}(\abs{\bm{v}_n})
	\Gamma \text{diag}(\abs{\bm{v}_m})\Big)\notag
\end{align}
where $\bm{\varphi}_n$  is re-centered  by $\bm{\varphi}^\top_n = [{\varphi}_{n_a}, {\varphi}_{n_b}, {\varphi}_{n_c}] \triangleq   [\angle{v}_{n_a}, \angle{v}_{n_b}+\frac{2\pi}{3}, \angle{v}_{n_c}-\frac{2\pi}{3}] $, and
$\Gamma$ is expressed as:
\begin{equation}
\begin{split}\small
 [\Gamma]_{k n} = e^{\mathfrak{j}\frac{2(k - n)\pi}{3}} = [\Gamma_c]_{kn}  + \mathfrak{j}[\Gamma_s]_{kn},~ k,n \in\{0,1,2\}\notag
\end{split}
\end{equation}
\begin{enumerate}
    \item In the active power equation the dominant term is:
\begin{align}
\bm{v}_n\bm{v}_{m}^H \approx \mathfrak{j}(\bm{\varphi}_n\mathbb{1}^\top \!\!\!-\! \mathbb{1}\bm{\varphi}_m^\top )\label{approx1}
\end{align}
\item In the reactive power equation the dominant term is:
\begin{align}
\bm{v}_n\bm{v}_{m}^H \approx \text{diag}(\abs{\bm{v}_n})
	\Gamma \text{diag}(\abs{\bm{v}_m})\label{approx2}
\end{align}
\end{enumerate}
Replacing $m$ with $n$ in \eqref{approx1} and \eqref{approx2}, we have the approximation for $\bm{v}_n\bm{v}_{n}^H$ for the active and reactive power injections.
\end{lemma}
\begin{proof} See the proof in the appendix.
\end{proof}

% {\color{red} Show the approximation and its proof in the appendix}
Applying the approximations in Lemma 1 to the expression in \eqref{sinj}  we obtain the physics in inspired GSO introduced in the the following Proposition:
\begin{proposition}
Let us define the following matrices:
\begin{eqnarray}
&\hat{\mathbf{B}}^{s}_{mn}\triangleq \Gamma_c \circ {\mathbf{B}}^{s}_{mn},
% %~~~~~~~ 
% &\hat{\mathbf{B}}^{(n)}_{mn}\triangleq \Gamma_c \circ {\mathbf{B}}^{(n)}_{mn}, \\
% &\hat{\mathbf{B}}^{(m)}_{mn}\triangleq \Gamma_c \circ {\mathbf{B}}^{(m)}_{mn}, 
%~~~~~~~ 
&\hat{\mathbf{B}} \triangleq ((\mathbb{1}\mathbb{1}^\top)_{N}  \otimes \Gamma_c)\circ \mathbf{B},\\
 &{\bm{p}}_n^{cst}\triangleq\frac{1}{2} D\left(\Gamma_s    \mathbf{B}^{s}_{mn}  \right) , %~~ 
& {\bm{q}}_n^{cst} \triangleq -\frac{1}{2}D(  \hat{\mathbf{B}}^{s}_{mn}) 
\end{eqnarray}
where $(\mathbb{1}\mathbb{1}^\top)_{N}$ is the all-ones matrix with dimension $N$ and $\otimes$ is Kronecker product.
Let ${\bm{p}}$, ${\bm{q}}$, ${\bm{p}^{cst}}$, ${\bm{q}^{cst}}$, $\abs{\bm{v}}$ and $\bm{\varphi}$  be the vectors stacking all the sub-vector corresponding to the multi-phase grid buses.
The following approximations holds:
\begin{equation}\label{GSOGS}
\begin{bmatrix}
{\bm{p}}  \\
{\bm{q}} 
\end{bmatrix}
-
\begin{bmatrix}
{\bm{p}}^{cst}  \\
{\bm{q}}^{cst} 
\end{bmatrix}
=
\overbrace{\begin{bmatrix}
\mathbf{\hat{B}}  & \bm 0  \\
\bm 0 & \mathbf{\hat{B}} 
\end{bmatrix}}^{\mathbf{S}:~\text{GSO}}
\overbrace{\begin{bmatrix}
\bm{\varphi} \\
\abs{\bm{v}} 
\end{bmatrix}}^{\bm x:~\text{GS}}=\mathbf{S}\bm x
\end{equation}
where $\bm{x}$ and $\bm S$ are graph filters in our problem.
%\begin{align}{\bm{p}} = \mathbf{\hat{B}} \bm{\varphi} + {\bm{p}}^{cst}, \\{\bm{q}} = \mathbf{\hat{B}} \abs{\bm{v}} + {\bm{q}}^{cst},
%\end{align}
% \begin{align}
% {\bm{p}} =  \begin{bmatrix}
% {\bm{p}}_1^\top,\cdots,
% {\bm{p}}_{\abs{\mathcal{N}}}^\top
% \end{bmatrix}^\top,
% %%%%
% \bm{\varphi} =  \begin{bmatrix}
% \bm{\varphi}_1^\top,\cdots,\bm{\varphi}_{\abs{\mathcal{N}}}^\top
% \end{bmatrix}^\top,\notag\\
% %%%%
% {\bm{q}} =  \begin{bmatrix}
% {\bm{q}}_1^\top,\cdots,
% {\bm{q}}_{\abs{\mathcal{N}}}^\top
% \end{bmatrix}^\top,
% %%%%
% \abs{\bm{v}} =  \begin{bmatrix}
% \abs{\bm{v}_1}^\top, \cdots ,
% \abs{\bm{v}_{\abs{\mathcal{N}}}}^\top
% \end{bmatrix}^\top,\notag\\
% %%%%
% {\bm{p}^{cst}} =  \begin{bmatrix}
% {\bm{p}^{cst}_1}^\top,\cdots,
% {\bm{p}^{cst}_{\abs{\mathcal{N}}}}^\top
% \end{bmatrix}^\top,
% %%%%
% {\bm{q}^{cst}} =  \begin{bmatrix}
% {\bm{q}^{cst}_1}^\top,\cdots,
% {\bm{q}^{cst}_{\abs{\mathcal{N}}}}^\top
% \end{bmatrix}^\top,\notag
% %%%%
% \end{align}
% and $\bm{\varphi}_n = [\varphi_{n_a}, \varphi_{n_b}, \varphi_{n_c}] \triangleq [{\varphi}_{n_a}, {\varphi}_{n_b}+\frac{2\pi}{3}, {\varphi}_{n_c}-\frac{2\pi}{3}]$.
% and $\hat{\mathbf{B}}$ is defined as
% \begin{equation}
% \hat{\mathbf{B}} = ((\mathbb{1}\mathbb{1}^\top)_{N}  \otimes \Gamma_c)\circ \mathbf{B}.
% \end{equation}

\end{proposition}
\begin{proof} See the proof in the appendix.
\end{proof}

 Note that the  GSO in \eqref{GSOGS} is a valid Laplacian matrix. We emphasize that this linearization is different from the existing DC linearization \cite{liu2018data} for the three-phase unbalanced power flow equations and that the Laplacian matrix $\mathbf{\hat{B}}$ is a modified  susceptance matrix ${\mathbf{B}}$.
%  Although the DC linearization still has the approximation error compared with AC power flow equations, the matrix polynomials of the GSO $\mathbf{S}$ and nonlinear activation functions (that will be introduced next) are utilized to further capture the spatial  correlations with high accuracy, which are validated in the case studies.

% \begin{figure}[htb]
% \centering
% \includegraphics[width=2.8in]{dynamic_graph.eps}
% \caption{Information flow of the DRL Agent training process.}\label{dynamic_graph}
% \vspace{-0.5cm}
% \end{figure}
\vspace{-0.3cm}
\subsection{GCN and GRN}
\vspace{-0.1cm}
Power systems are dynamic systems with time-varying voltage phasors. 
%In this subsection, we consider the temporal correlations of graphs, as shown in Fig. \ref{dynamic_graph}.
In order to fuse features from both spatial and temporal domains, we introduce GCN  and GRN.

\subsubsection{GCN}
Based on \eqref{stkernal}, we can design the following transfer functions:
\begin{equation}
\begin{split}\label{stkernal2}
\tilde{\mathbb{H}}(\lambda, z) =\sum_{k=0}^K  {\Theta}_{k} \lambda^k \left(  \sum_{t=0}^{T-1}   h_{k, t}  z^{-t} \right),
\end{split}
\end{equation}
where $\left(  \sum_{t=0}^{T-1}   h_{k, t}  z^{-t} \right)$ are the parameters of  temporal convolution, and then the GCN blocks are utilized. Accordingly, the graph signal $\bm w^c_t$ from the first feature extraction layer   is:
\begin{equation}
\vspace{-0.1cm}
\begin{split}\label{graph_signal}
{\bm w^c_t} = \operatorname{ReLU} \left[\sum_{k=0}^K  \Theta_{k, 1} \mathbf{S}^k \left(  \sum_{\tau=0}^{T-1}   h_{k, t}  {\bm x_{t-\tau} } \right)\right].
\end{split}
\end{equation}
Followed by  the feature extraction  layer \eqref{graph_signal},
the remaining hidden layers $\ell\in\{1, \cdots, L-1\}$    are analogous to those of a  fully connected neural network:
\begin{equation}
\vspace{-0.1cm}
\begin{split}\label{cgcn_policy}
%{\bm w^2_t} = \sigma\left(\Theta_2 \cdot {\bm w^1_t}\right),  \\
&{\bm{w}}^c_{t,\ell+1}=  {\operatorname{ReLU}\left(\Theta_\ell  \cdot {{\bm{w}}^c_{t,\ell}}\right)}, \quad L-1\ge\ell \ge 1.
%  &   \bm{y} = \sigma\left(\Theta_{L}^{real} \cdot [\Re(\bar{\mathbf{W}}_{t,L}), \Im(\bar{\mathbf{W}}_{t,L}) ]\right).
\end{split}
\vspace{-0.1cm}
\end{equation}
For the output layer $L$, the regression samples are:
\begin{equation}\label{cgcn_policy}
\vspace{-0.1cm}
\begin{split}
 ~\bm{y}_t = \text{tanh}\left(\Theta_{L}  \cdot 
{\bm{w}}^c_{t,L}
\right),
\end{split}
\vspace{-0.1cm}
\end{equation}
where $\bm{y}_t$ is the  regression targets  and $\Theta_{\ell}, \forall \ell = 1,\cdots, L $   is the   trainable matrix.
Finally, the multi-layer  GCN learning function is: 
\begin{equation}
\begin{split}
  \bm{y}_t = \Phi^c(\mathbf{X}_t, \mathbf{S}, \theta ),\label{predictlabel}
\end{split}
\end{equation}
where $\theta  \triangleq \{(\Theta_{\ell}, \Theta_{k}, h_{k, t})|\forall \ell, \forall k\}$   represent the trainable parameters and $\mathbf{X}_t = [\bm{x}_{t-T+1}, \cdots, \bm{x}_{t}]$. Here, we have omitted the bias term   to unburden the notation, but they are present in the trainable model we use. 

\subsubsection{GRN}
RNNs are systems that exploit recurrence to learn dependencies in sequences of variable length.
Next, we adapt the operations performed by RNNs to take the graph structure into account when dealing with graph processes as follows:
%. Based on \eqref{stkernal}, we have
\begin{equation}
\begin{split}\label{graph_signal_rnn}
{\bm w^r_t} = \operatorname{ReLU}\left(\Theta_{t} \left(\sum_{k=0}^K  {h}_{k} \mathbf{S}^k {\bm x_{t} }\right)+  \Theta_{t-1} {\bm w^r_{t-1} }\right),
\end{split}
\end{equation}
where we also omit  the biased term to unburden the notation. 
The remaining hidden and output layers layers of GRN are:
\begin{equation}
\begin{split}\label{rgcn_policy}
%{\bm w^2_t} = \sigma\left(\Theta_2 \cdot {\bm w^1_t}\right),  \\
&{\bm{w}}^r_{t,\ell+1}=  {\operatorname{ReLU}\left(\Theta_\ell  \cdot {{\bm{w}}^r_{t,\ell}}\right)}, \quad L-1\ge\ell \ge 1,\\
& \bm{y}_t = \text{tanh}\left(\Theta_{L}  \cdot 
{\bm{w}}^r_{t,L}
\right),   
%  &   \bm{y} = \sigma\left(\Theta_{L}^{real} \cdot [\Re(\bar{\mathbf{W}}_{t,L}), \Im(\bar{\mathbf{W}}_{t,L}) ]\right).
\end{split}
\end{equation}
where  $\bm{y}_t = \Phi^r(\mathbf{X}_t, \mathbf{S}, \theta)$ is defined as in the multi-layer GRN, $\theta  \triangleq \{(\Theta_{\ell}, \Theta_{t}, h_{k, t})|\forall \ell, \forall k\}$ represent the trainable parameters.  %Here, we omit the bias term for brevity, but it is added in the implementation.   
%
%In particular, we have these parameters, i.e., $h_k$ in GCN and $[\Theta_x,\Theta_h]$ within RNN, to be learnt by the DRL framework. {\color{red}Likewise, these parameters responds to policy parameters $\theta$ of $\pi_{\theta}(\cdot|\bm{x}_t)$ in Eq. \eqref{eqobj} (if we only consider the RGCN layer).}
Similar to our description of GCN, we illustrate the proposed GRN with $K=3$ on the middle side of Fig. \ref{fig:CGCN_RGCN}.  
It shows that the graph signals are processed through GCN to capture spatial features of the signal, and then processed by an RNN to capture the temporal correlation.  Both GCN and GRN architectures capture the spatio temporal correlation of voltage phasors, but they have unique advantages. In \eqref{cgcn_policy},  the proposed GCN model  handles the time convolutions  via a CNN that allows to use GPUs to accelerate the computations during training.
In contrast, the GRN has long   memory built in due to the
feedback connections with ${\bm w^r_t}$ and ${\bm w^r_{t-1}}$ in \eqref{graph_signal_rnn}. Hence, GRN is best suited for environments driven by state equations.

\vspace{-0.3cm}
\subsection{GSO for Partial Observation}\label{sec:reducedGSO}
\vspace{-0.1cm}
It is very useful for GCN architectures to be able to accept an input that does not include the complete information about the state, because of lacking measurements or for better scalability. 
In this subsection, we provide the correct GSO for a down-sampled graph signal as an input of a reduced order GCN.
% {\color{red}Move it to Section III.}
%In this subsection, we  introduce the low-pass property of down-sampled graph signal, and then design the  graph filter for the down-sampled graph signals. 
%After that, we present a PMU selection algorithm to enable the sparse observation of voltage phasor measurements in order to solve the problem of scalability.  
%Regarding the situation of unavailable PMU measurements in some microgrids, we design an neural network mapping from AMIs to voltage phasors. 
%Equipped with this mapping function, \eqref{GSOGS} can be extended to capture the spatio-temporal correlations of AMIs.
% \subsubsection{Low-Pass Property of Down-Sampled Graph Signal}
Let $\bm v_{\mathcal{M}}$ (time index $t$ is ignored for simplicity) be the down-sampled voltage graph signal where $\mathcal{M} \in \mathcal{N}$ in the set of node indices of the corresponding buses. 
%
%In \cite{ramakrishna2021gridgraph}, it was shown that, not only the entire system state, but any down-sampled voltage phasor vector obtained from it has still low-pass property (its GFT is concentrated in the low graph frequencies) with respect to a GSO that is the Kron reduced susceptance matrix. 
We leverage the result is \cite{ramakrishna2021gridgraph} summarized in the following lemma:
\begin{lemma}[{{\cite[Lemma 1]{ramakrishna2021gridgraph}}}]\label{lemma1} To define the GSO with respect to the reduced-graph of $\mathcal{M}$, denoted by $\mathbf{S}_{red, \mathcal{M}}$
let us partition the grid GSO  $\mathbf{S}$ as follows: 
\begin{equation}
\begin{split}\label{GSO_block}
\bf{S} = \begin{bmatrix}
\bf{S}_{\mathcal{M}, \mathcal{M}} & \bf{S}_{\mathcal{M}, \mathcal{M}^c} \\
\mathbf{S}^\top_{\mathcal{M}, \mathcal{M}^c} & \bf{S}_{\mathcal{M}^c, \mathcal{M}^c}
\end{bmatrix}.
\end{split}
\end{equation}
From Ohm's law shows it follows that the samples of the state $\bm v_{\mathcal{M}}$ are such that:
\vspace{-0.2cm}
%the output of low-pass graph filter  $\mathcal{H}(\mathbf{S}_{red, \mathcal{M}}) \triangleq  \mathbf{S}^{-1}_{red, \mathcal{M}}$
\begin{equation}
\begin{split}\label{reduced_GSO}
\bm{v}_{\mathcal{M}} = \mathcal{H}(\mathbf{S}_{red, \mathcal{M}}) \bm{\kappa},
\end{split}
\end{equation}
where $\mathbf{S}_{red, \mathcal{M}}$ is the Schur complement of the block  $\mathbf{S}_{\mathcal{M}^c \mathcal{M}^c}$ (which is the Kron-reduction of $\bf{S}$), i.e.,
%\vspace{-0.2cm}
\begin{equation}
\begin{split}\label{GSO_reduced}
\mathbf{S}_{red, \mathcal{M}}  = \mathbf{S}_{\mathcal{MM}} - \mathbf{S}_{\mathcal{MM}^c} \mathbf{S}^{-1}_{\mathcal{M}^c \mathcal{M}^c} \mathbf{S}^\top_{\mathcal{MM}^c}.\notag
\end{split}
\end{equation}
%, $\mathbf{S}_{red, \mathcal{M}} = \mathcal{SC}(\bf{S}, \bf{S}_{\mathcal{M}^c, \mathcal{M}^c})$. The $\bf{S}_{\mathcal{M}^c, \mathcal{M}^c}$ is a sub-matrix of $\bf{S}$ as:
%\vspace{-0.2cm}
%\begin{equation}
%\begin{split}\label{GSO_block}
%\bf{S} = \begin{bmatrix}
%\bf{S}_{\mathcal{M}, \mathcal{M}} & %\bf{S}_{\mathcal{M}, \mathcal{M}^c} \\
%\mathbf{S}^\top_{\mathcal{M}, \mathcal{M}^c} & %\bf{S}_{\mathcal{M}^c, \mathcal{M}^c}
%\end{bmatrix}.
%\end{split}
%\end{equation}
%\vspace{-0.2cm}
% Recall that a graph signal $\bm x$ with arbitrary graph frequency response with respect to GSO $\mathbf{S}$: 
% \begin{equation}
% \begin{split}\label{reGSO}
% \bm{x} =\mathbf{S}^{-1} \bm{w} = \mathbf{S}^{-1} (\mathbf{S} (\mathbf{S}^{-1}) \bm{w}).
% \end{split}
% \end{equation}
\end{lemma}
\begin{proof}
The statement follows from the observation that \eqref{reduced_GSO} holds since from Ohm's law:
\begin{equation}
\begin{split}
\bm{v}_{\mathcal{M}}= \overbrace{\mathbf{S}^{-1}_{red, \mathcal{M}}}^
{\mathcal{H}(\mathbf{S}_{red, \mathcal{M}})}
\overbrace{[\mathbb{I}_{\abs{\mathcal{M}}} - \mathbf{S}_{\mathcal{M} \mathcal{M}^c} \mathbf{S}^{-1}_{\mathcal{M}^c \mathcal{M}^c}](\mathbf{S} (\mathbf{S}^{-1}) \bm{i})}^{\bm \kappa},
\end{split}
\end{equation}
where $\mathbb{I}$ is an identity matrix. The equation establishes a generative graph filter model, with GSO $\mathbf{S}_{red, \mathcal{M}}$ for the decimated voltage phasors, supporting such GSO choice.  
\vspace{-0.2cm}
\end{proof}

\vspace{-0.2cm}
\section{GCN and GRN Applications}
Through two applications, in this section we illustrate  how our framework can be used to advance AI for grid data. 
%This section utilize the proposed STGCN for two applications in power systems, i.e., PSSE and PSSF, and DRL control of smart inverters. For the former one, the proposed STGCN could be applicable to both three-phase  distribution networks and transmission networks. For the latter one, due to the fact that smart inverters are usually installed in the three-phase distribution networks, we embed proposed STGCN in a  DRL framework for the distribution network voltage control.   The proposed STGCN is general so that it could be embedded in any state-of-art reinforcement learning frameworks. %

\vspace{-0.2cm}
\subsection{Power System State Estimation and Forecasting} 
An important contribution of our design is its capability to take inputs that contain only a subset ${\cal M}$ of state variables. It is natural to expect that the performance of the PSSE and PSSF is affected by the subset $\mathcal{M}$ where PMUs are installed. 
%To enable the multi-agent distributed DRL, we also propose a power grid clustering method based on GSP to decompose distribution networks. Definition \ref{def:eta} indicates the low-pass signals have  a cult-off frequency $\lambda_k$ such that frequency content corresponding to $\lambda_{k+1}$ and higher is negligible  \cite{ramakrishna2021gridgraph}. 
%Nontheless, the community structure of the graph allows for reducing the dimensionality of the graph signal by sampling.
 This is why  we provide an optimized criterion to select ${\cal M}$ leveraging GSP sampling theory. 
Let the GFT basis corresponding to the first dominant $k$ graph frequencies be $\mathbf{U}_{\mathcal{K}}$. As shown in \cite{ramakrishna2021gridgraph} the best ${\cal M}$ is one-to-one with the subset of rows of $\mathbf{U}_{\mathcal{K}}$ with minimum correlation.
Let ${\mathcal{F}}_\mathcal{M}$ be what is called the {\it vertex limiting operator} i.e. the matrix such that ${\mathcal{F}}_\mathcal{M} = \bm{\mathcal{Q}}_\mathcal{M} \bm{\mathcal{Q}}^\top_\mathcal{M}$, where $\bm{\mathcal{Q}}_\mathcal{M}$ has columns that are the coordinate vectors pointing to each vertex/node in ${\cal M}$. 
%Therefore, we have $[\bm{x}_t]_{\mathcal{M}} = \bm{\mathcal{Q}}_\mathcal{M}^\top \bm{x}_t$. For reconstruction, the necessary condition should be that $\abs{\mathcal{M}} \ge {\abs{\mathcal{K}}}$.
%Therefore, the best placement of PMUs on the grid to minimize the worst-case reconstruction error is closely tied to the grid topology and the model mismatch relative to a strictly band-limited graph signal \cite{anis2016efficient}.  
Mathematically, the optimal placement can be sought by maximizing the smallest  singular value, $\max_{{\mathcal{F}}_{\mathcal{M}}}\varpi_{\min} ({\mathcal{F}}_{\mathcal{M}} \mathbf{U}_{\mathcal{K}})$, of the matrix ${\mathcal{F}}_{\mathcal{M}} \mathbf{U}_{\mathcal{K}}$. Such choice amounts to the selection of rows of $\mathbf{U}_{\mathcal{K}}$ that are as uncorrelated as possible, because the resulting matrix ${\mathcal{F}}_{\mathcal{M}} \mathbf{U}_{\mathcal{K}}$ has the highest conditional number \cite{anis2016efficient}.

After choosing the best location of PMUs $\mathcal{M}$ by the aforementioned method, we have  the sub-sampled measurement $[\bm{z}_t]_{\mathcal{M}}$.  
%Assume that we have access to  measurements of voltage and current phasors from a subset of buses with PMUs installed.
%
Let $\mathcal{M}$ denote the set of available measurement buses  and $\mathcal{U}$ denote the set of  unavailable ones. Therefore, \eqref{eq:node_injection} can be written as:
% \vspace{-0.1cm}
\begin{equation}
\underbrace{\left[\begin{array}{l}
\hat{\boldsymbol{i}}_{\mathcal{M}} \\
\hat{\boldsymbol{v}}_{\mathcal{M}}
\end{array}\right]}_{\boldsymbol{z}_{t}}=\underbrace{\left[\begin{array}{cc}
\boldsymbol{Y}_{\mathcal{M M}} & \boldsymbol{Y}_{\mathcal{M} \mathcal{U}} \\
\mathbb{I}_{|\mathcal{M}|} & \mathbf{0}
\end{array}\right]}_{\boldsymbol{H}} \underbrace{\left[\begin{array}{l}
\boldsymbol{v}_{\mathcal{M}} \\
\boldsymbol{v}_{\mathcal{U}}
\end{array}\right]}_{\boldsymbol{x}_t}+\boldsymbol{\varepsilon}_{t},
\end{equation}
where $\boldsymbol{\varepsilon}_{t}$ is a vector of measurement noise.
Our task is to  estimate  the  voltage phasors at the present time and forecast the future voltage phasors by the GCN and GRN methods. The time-series  voltage phasor forecasting problem is  modeled as predicting the most likely voltage phasors in the next $H$ time steps given the  previous  $T$ sub-sampled observation $[\bm{x}_t]_{\mathcal{M}}$ as
\begin{equation}
\begin{aligned}
{\bm x}'_{t+H}=  \underset{ \bm{x}_{t+H}}{\arg \max } \log P\left( \bm{x}_{t+H} \mid \bm{z}_{t-T+1}, \ldots, \bm{z}_{t}\right),\notag
\end{aligned}
\vspace{-0.1cm}
\end{equation} 	
where $[\bm{z}_{t}] \in\mathbb{R}^{2\abs{\mathcal{M}}}$ is an observation vector of $ \abs{\mathcal{M}}$  measurements for both voltage and current phasors at time step $t$, each element of which records the historical observation for a bus.

\subsubsection{Methodology}
First of all, we need to recover the voltage phasors $\bm{x}_t$ based on $\mathcal{M}$ measurements, i.e. $\bm{z}_t = \left[\hat{\boldsymbol{i}}_{\mathcal{M}},\hat{\boldsymbol{v}}_{\mathcal{M}} \right]^\top$ by solving the regularized least square problem:
\begin{equation}
   \min_{\bm{x}_t} \norm{\bm{z}_t  - \bm{H}\bm{x}_t}^2_2+ \mu_1 (\bm{x}_t^H %\left[\begin{array}{cc}
%\boldsymbol{Y}_{\mathcal{A A}} & \boldsymbol{Y}_{\mathcal{A} \mathcal{U}} \\
%\boldsymbol{Y}_{\mathcal{U} \mathcal{A}} & \boldsymbol{Y}_{\mathcal{U U}}
%\end{array}\right]
\mathbf{S}
\bm{x}_t)\label{RLS}
\end{equation}
where $\mu_1$ is positive. The closed-form solution of \eqref{RLS} is:
\begin{equation}
  \hat{\bm{x}}_t = \left(\bm{H}^H\bm{H}  + \mu_1
  %\left[\begin{array}{cc}
%\boldsymbol{Y}_{\mathcal{A A}} & \boldsymbol{Y}_{\mathcal{A} \mathcal{U}} \\\boldsymbol{Y}_{\mathcal{U} \mathcal{A}} & \boldsymbol{Y}_{\mathcal{U U}}
%\end{array}\right]
\mathbf{S}
\right)^\dagger\bm{H}^H\bm{z}_t, \label{recover_vp}
\end{equation}
where $\hat{\bm{x}}_t$ is the estimated voltage phasor and where  $(\cdot)^\dagger$ denotes the pseudo-inverse.
%  Voltage data samples are obtained downsampling in space after the optimal placement of PMUs. At time $t$, with $\abs{\mathcal{M}}$ samples, $[\bm{x}_t]_{\mathcal{M}}$ are available,  and the following model applies
% \begin{equation}
%     [\bm{x}_t]_{\mathcal{M}} \approx  {\mathcal{F}}_\mathcal{M}^\top \mathbf{U}_{\mathcal{K}} \tilde{\bm{x}}_t
% \end{equation}
% where $\tilde{\bm{x}}_t$ is the GFT of graph signal $\bm{x}_t$. Therefore, reconstruction in spatial domain is done via GFT basis as
% \begin{equation}
%     \hat{\bm{x}}_t = \mathbf{U}_{\mathcal{K}}( {\mathcal{F}}_\mathcal{M}^\top  \mathbf{U}_{\mathcal{K}})^\dagger     [\bm{x}_t]_{\mathcal{M}}.\label{reconstru}
% \end{equation}
% where  $(\cdot)^\dagger$ denotes the pseudo-inverse.
The complete algorithm is below: 
  \begin{algorithm} [!htb]\small
      \caption{Voltage Phasor    Forecasting}
  We collect $T$ historical sub-sampled measurements $ \bm{z}_{t-T+1} , \ldots,  \bm{z}_{t}  $\;
    We utilize \eqref{recover_vp} to obtain the estimated full observations $\hat{\bf{X}} = [\hat{\bm{x}}_{t-T+1}, \ldots, \hat{\bm{x}}_{t}]$\;
 The loss function of the  GCN or GRN function for voltage phasor prediction is written as
    \begin{align}
     &\mathcal{L}(\Phi, \theta)  = \sum_t \Big\{\norm{\bm{y}_t - \bm{x}_{t+H} }^2 +\label{objfore}\\
     & \mu_2 \norm{ \hat{\boldsymbol{v}}_{t+H,\mathcal{M}}  \circ \hat{\boldsymbol{i}}_{t+H, \mathcal{M}}^*  - \big[\bm{y}_t \circ  (\mathbf{S} \bm{y}_t)^*\big]_{\mathcal{M}}}^2 \Big\},\notag
    \end{align}
    where $\bm{x}_{t+H}$ is the ground truth voltage phasor in the next $H$ time step, $(\cdot)^*$ denotes the conjugate operator and $\bm{y}_t = \Phi(\hat{\mathbf{X}}_t, \mathbf{S}, \theta ) $  is the predicted target to approximate the ground-truth regression target $\Re(\bm{x}_{t+H})$ and $\Im(\bm{x}_{t+H})$, where $\Phi$ could be either $\Phi^c$ in \eqref{predictlabel} or $\Phi^r$ in \eqref{rgcn_policy}. Note that $H = 0$ is the voltage phasor estimation, and $H \ge 1$ is the voltage phasor forecasting\;
\end{algorithm}

Note that the regularization term in \eqref{objfore} favors voltage phasor forecasts that minimize the sum of the absolute value of the apparent power injections.
After   training,  we use $\Phi(\hat{\mathbf{X}}_t, \mathbf{S}, \theta)$ to   forecast  $\bm{x}_{t+H}$ given the observations $ \bm{z}_{t-T+1} , \ldots,  \bm{z}_{t} $.

\vspace{-0.2cm}
\subsection{Deep Reinforcement Learning Control of Smart Inverters}
Following \cite{zhang2020deep, jha2019bi} we consider reactive power support for smart inverters operating on a per-phase basis, where each smart inverter is installed at bus $n_\phi$.
%and the capacity (maximum active power), denoted by $s_{n_\phi}$, of the Photo-Voltaic (PV) panel is known by a short term forecasting.
The reactive power support of the smart inverter depends on the nominal per-phase capacity, denoted by $s_{n_\phi}$. Specifically, the range of possible reactive power, $q_{n_\phi}$ for the smart inverter is:
\begin{equation}
\begin{split}\small
|{q}_{n_\phi}|\leq \bar{q}_{n_\phi}\triangleq
\sqrt{s_{n_\phi}^2 -   p_{n_\phi}^2}
\end{split}
\end{equation}
where $\bar{q}_{n_\phi} =  \sqrt{s_{n_\phi}^2 -   p_{n_\phi}^2}$  denotes the maximum reactive power of this smart inverter installed at bus ${n_\phi}$. Here, we define the control variable as $a_{n_\phi}\in [-1, 1]$, and the reactive power injected into the distribution network is ${q}_{n_\phi} = a_{n_\phi}  \bar{q}_{n_\phi}$. 
\subsubsection{Methodology}
The DRL training is aimed at learning the parameters of GCN (or GRN) encoding the optimum stochastic policy function for mapping the voltage phasors measurements onto the control variables $a_{n_\phi}$ that regulate the voltage magnitude. 
%The GCN or GRN trained by the DRL algorithm is the agent selecting specific inverter action $a_{n_\phi}$, and the environment, where the agent takes multiple control actions, is  the electric distribution network. 
%The policy must respond to a variety of conditions and take control actions with respect to the given operating condition to achieve Volt-VAR control.
% In a  DRL process, a neural network can output multiple  actions to control the power injections.
%However, the dimension of the action space increases explosively with the number of controllable devices installed in the three-phase distribution system. In contrast,  a multi-agent DRL has multiple actors that could control either one smart inverter or multiple smart inverter based on local observations. 
The interaction between the agent and the grid environment at time $t$ is described by: 1) the state, comprising a set of past samples $ (\bm x_t,\ldots,\bm x_{t-T+1})$, 2) the action $\bm a_t$, and 3) the reward $\bm r_t$. 
%We describe these three elements next.
\paragraph{State and Action}
The tuple of actions for the Volt-VAR control of the inverters in the bus set ${\cal N}_s$ is denoted by:
\begin{equation}\label{eqobj}
\begin{split}
&\bm{a} = [a_1,  \cdots,a_{n_\phi}, \cdots, a_{\abs{\mathcal{N}_s}}], a_{n_\phi}\in   [-1, 1],
\end{split}
\vspace{-0.2cm}
\end{equation}
 and thus the corresponding reactive power injections are:
\begin{equation}
\begin{split}
&\bm{q}_{n_\phi} =[q_1, \cdots,q_{n_\phi}, \cdots, q_{\abs{\mathcal{N}_s}}], \text{where}~{q}_{n_\phi} =  a_{n_\phi}  \bar{q}_{n_\phi},
\end{split}
\vspace{-0.2cm}
\end{equation}
where $\abs{\mathcal{N}_s}$ represents the number of smart inverters, $a_i$ denotes the control action on the $i$th smart inverter.  
 The vector $\bm{a}$, output of the GCN (or GRN) approximating the optimum policy, is a function of the  three-phase voltages at all, or part, of the buses of the distribution system, which constitute the  observation and are the input of the GCN-DRL or GRN-DRL. The state vector/observation is $\bm{x} =[ {\bm{\varphi}};
\abs{\bm{v}}]^\top$, where ${\bm{\varphi}}$ is the vector of re-centered voltage phases and $\abs{\bm{v}} $ is the vector of voltage magnitudes. 
% \begin{figure}[htb]
% \centering
% \includegraphics[width=3.5in]{eg_action.pdf}
% \caption{Action Example.}\label{action_example}
% \vspace{-0.3cm}
% \end{figure} 

\paragraph{Reward}
The regret is defined as the magnitude of voltage  deviation from the reference  at bus $n_\phi$ at time $t$ as follows \cite{cao2020multi, xu2019optimal}:
\begin{equation}
    \begin{split}
        r_{n_\phi, t} = -\left|\abs{{v}_{n_\phi, t}} -\bar{v}  \right|, n_{\phi}\in \mathcal{N}_s,
    \end{split}
\end{equation}
where $\bar{v}$ denotes the desired voltage magnitude (i.e., 1 p.u.), and $\abs{{v}_{n_\phi, t}}$ is the measured voltage magnitude on phase $n_\phi$. % Note that $\{\abs{\bar{v} - \abs{{v}_{n_\phi, t}}} \} = \emptyset$ if there is no inverter deployed at bus $n_\phi$.

\paragraph{Objectives of DRL}
In this application, $\bm{a}_t = \Phi({\mathbf{X}}_t, \mathbf{S}, \theta ) $ denotes a stochastic policy  that models the probability distribution of $\bm{a}_t \in \mathcal{A}$ given a sequence of   observations  ${\mathbf{X}}_t$.
The goal of each agent is to find a policy, which maximizes its expected discounted return:
\begin{equation}\label{eqobj}
\begin{split}
\Phi({\mathbf{X}}_t, \mathbf{S}, \theta)  \in \argmax J(\pi)= \mathbb{E}_{\varsigma \sim \Phi}\left[\sum_{t=0}^{T} \gamma^\top r_t\right],
\end{split}
\end{equation}
where $\varsigma $ is the trajectory generated by policy $\Phi({\mathbf{X}}_t, \mathbf{S}, \theta)$, i.e., the action  $\bm{a}_t$ is taken according to policy  $\Phi({\mathbf{X}}_t, \mathbf{S}, \theta)$, $r_t$ represents rewards at time $t$. The parameter $\gamma \in  (0, 1)$ is the discounting factor, discounting future rewards.

% \vspace{-0.3cm}

\section{Experimental results}
In this section, we perform numerical experiments adopting the IEEE 118-bus transmission network and the 123-bus feeder distribution system to validate the proposed GCN and GRN  frameworks for power system state estimation and forecasting, and DRL-based voltage control. In both cases, for the formulation and training of the GCN and GRN architectures we relied on PyTorch.
\begin{figure}[!htb]
  \vspace{-0.3cm}
  \centering
  \subfigure[Voltage Magnitudes ($H=0$).]{
 \includegraphics[width=1.5in]{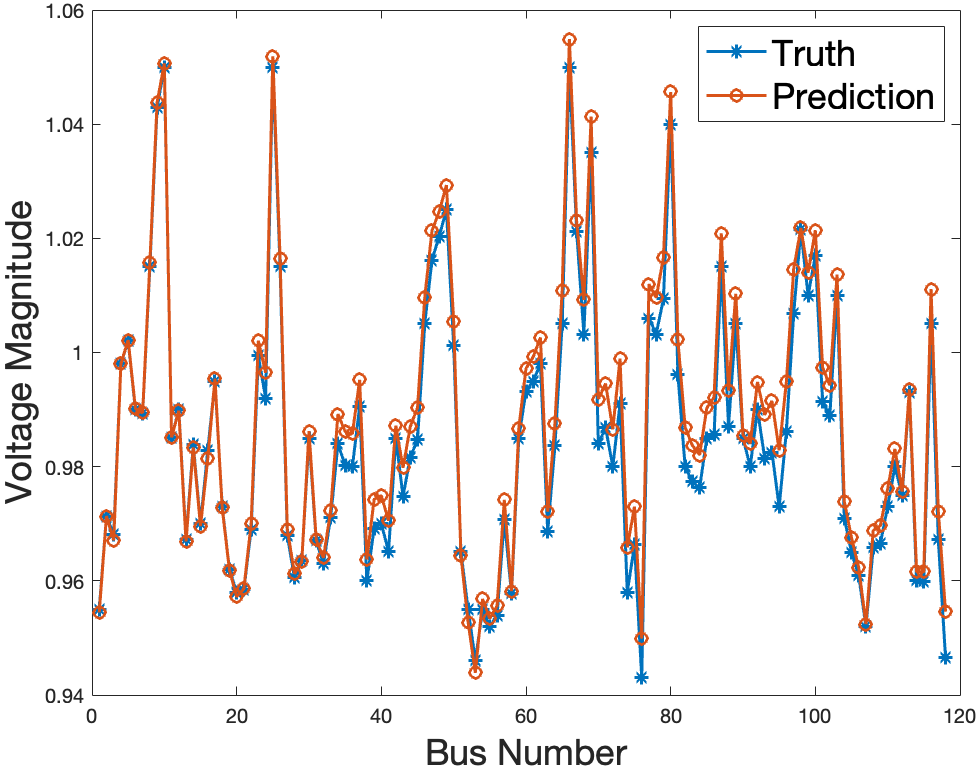}
 % \fbox{\rule[-.5cm]{0cm}{6cm} \rule[-.5cm]{6cm}{0cm}}
      \label{Figure1a}
 } 
\hspace{-0.1in}
 \subfigure[Voltage Phases ($H=0$).]{
\includegraphics[width=1.5in]{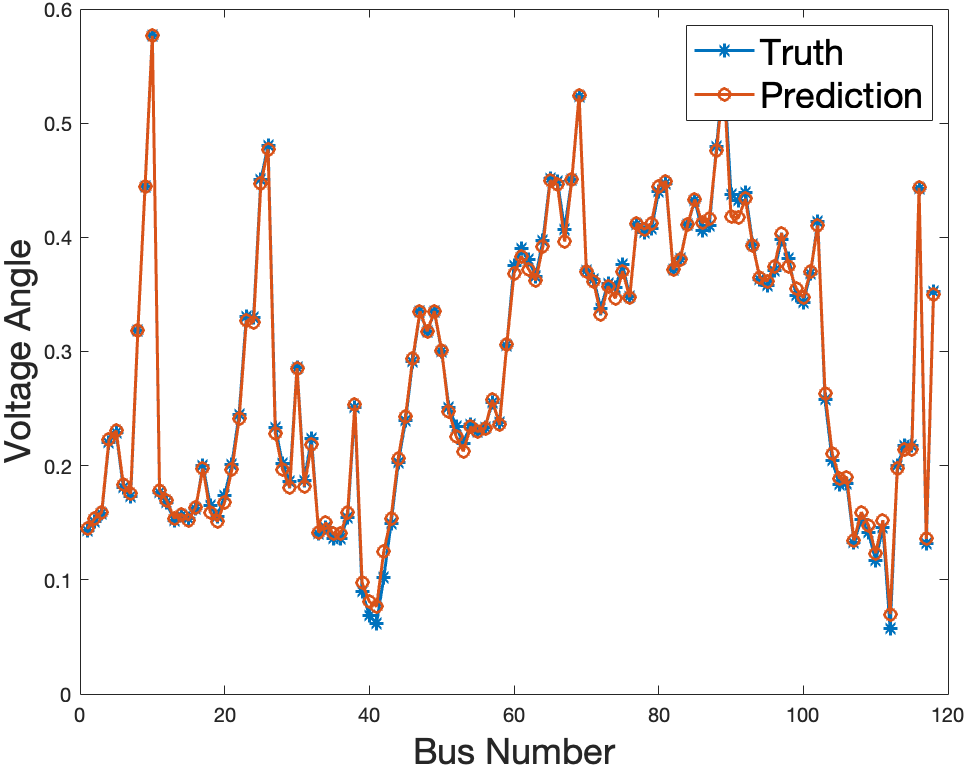}
% \fbox{\rule[-.5cm]{0cm}{6cm} \rule[-.5cm]{6cm}{0cm}}
     \label{Figure1a}
} 
  \subfigure[Voltage Magnitudes ($H=1$).]{
 \includegraphics[width=1.5in]{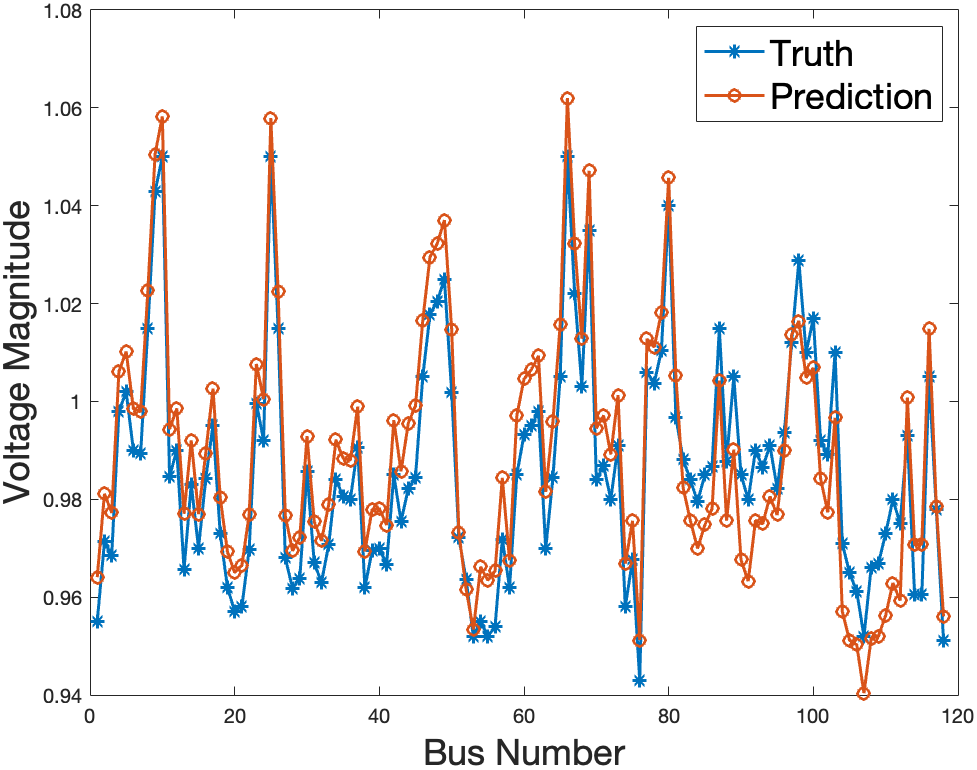}
 % \fbox{\rule[-.5cm]{0cm}{6cm} \rule[-.5cm]{6cm}{0cm}}
      \label{Figure1a}
 } 
\hspace{-0.1in}
 \subfigure[Voltage Phases ($H=1$).]{
\includegraphics[width=1.5in]{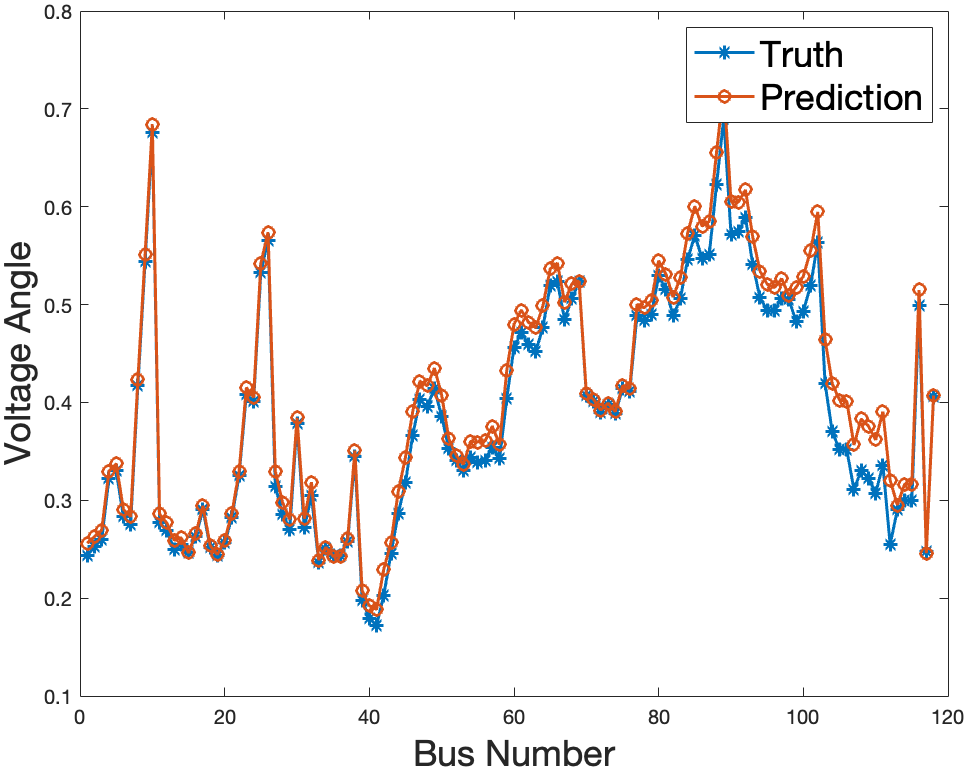}
% \fbox{\rule[-.5cm]{0cm}{6cm} \rule[-.5cm]{6cm}{0cm}}
     \label{Figure1a}
} 
  \caption{An example of PSSE and PSSF by GCN for the IEEE 118-bus system.}
  \vspace{-0.7cm}\label{estimation_forecast}
\end{figure}  
\vspace{-0.3cm}

\begin{figure*}[!htb]
	\centering
 \subfigure[Training curves of 3 small inverters based on full observations.]{
 \includegraphics[width=2.20in]{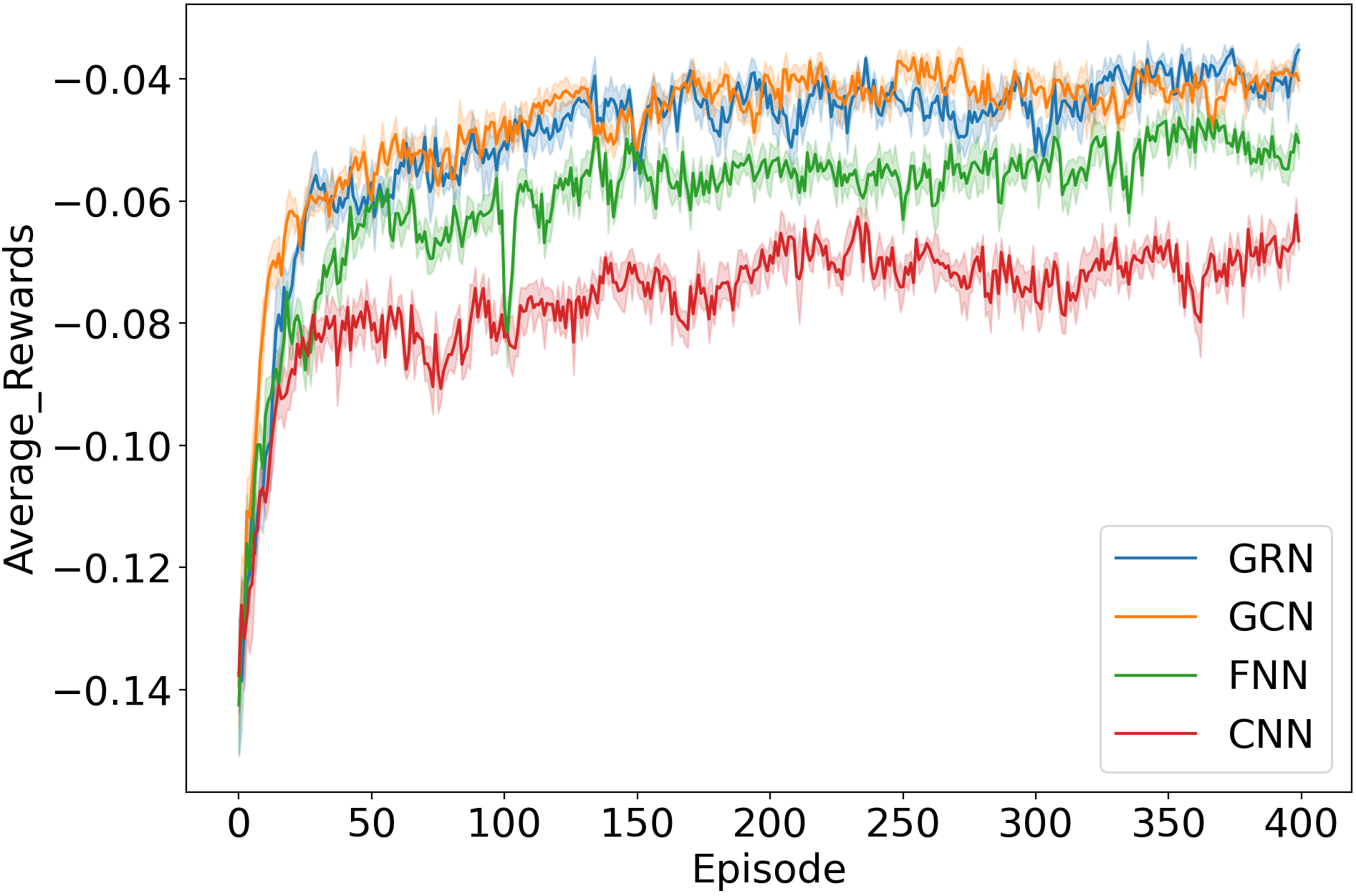}
      \label{Figure4a}
 }
%  \hspace{-0.18in}
 \subfigure[Training curves of 3 small inverters based on partial observations.]{
 \includegraphics[width=2.20in]{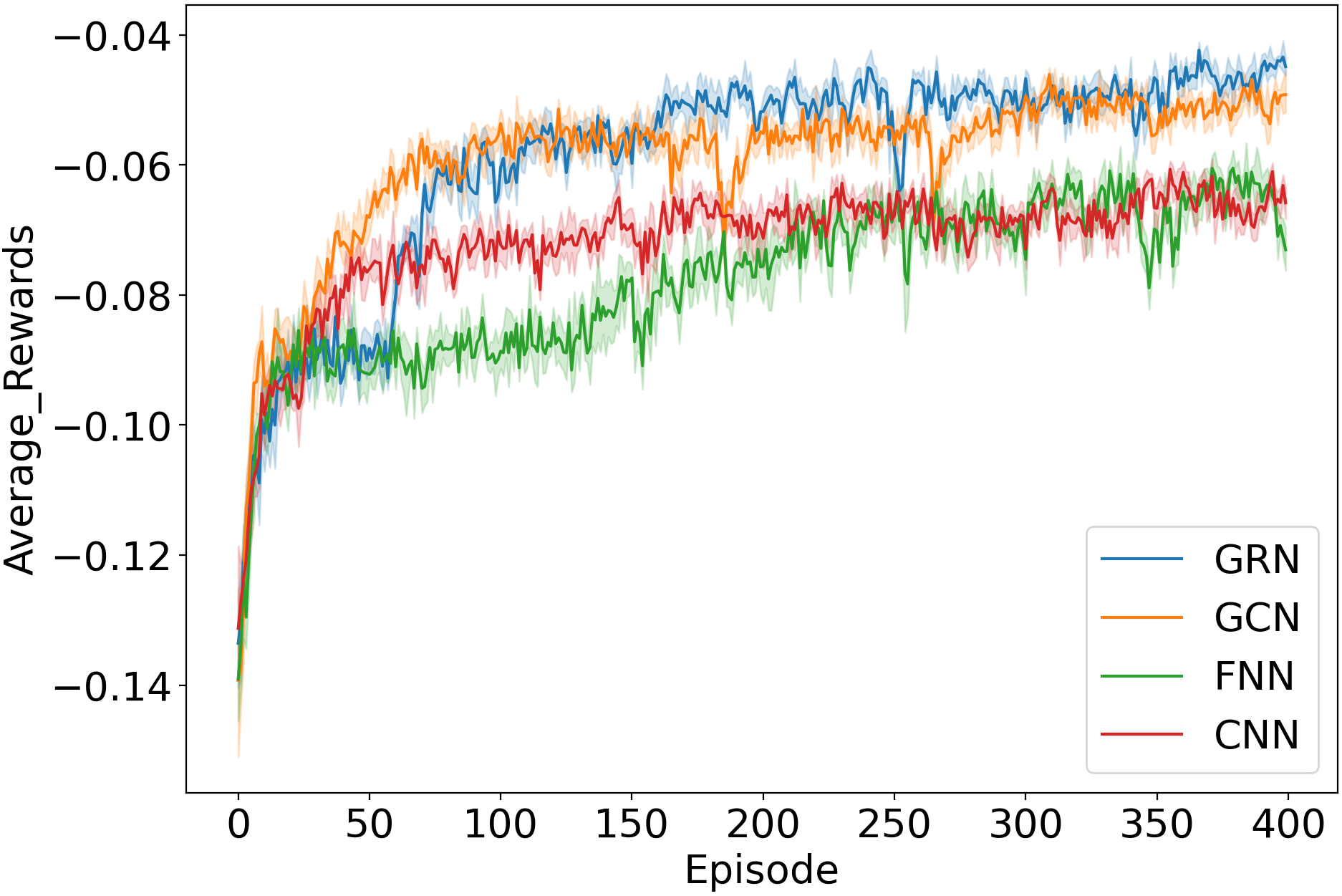}
   \label{Figure4b}
 }
 \subfigure[Training curves of 6 small inverters based on full observations.]{
\includegraphics[width=2.20in]{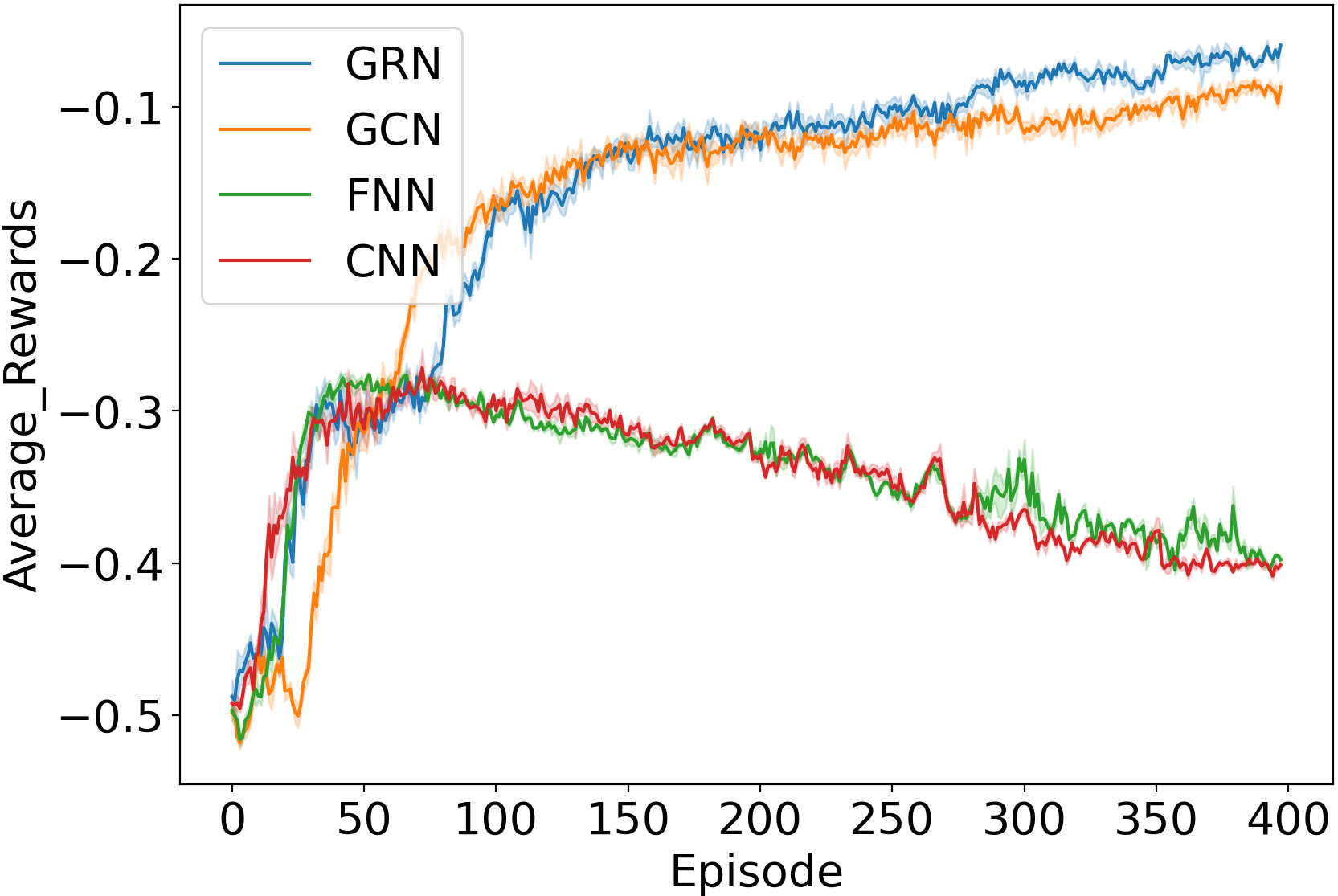}
     \label{Figure4c}
}
% \hspace{-0.18in}
% \subfigure[Training curves with $\zeta_v = 0$.]{
% \includegraphics[width=1.68in]{Figure4d.png}
%   \label{Figure4d}
% }
\vspace{-0.35cm}
	\caption{(a) and (b) have 3  smart inverters and (c) has  6 smart inverters. (a) and (b) illustrate the learning (training) curves of the GCN-DRL and GRN-DRL algorithms for voltage magnitude regulations with full and partial observations, respectively. }
	\label{Figure4}
	\vspace{-0.5cm}
\end{figure*} 
\subsection{Power System State Estimation and Forecasting}
\subsubsection{Experimental setup}
For the first application, we use  realistic load time-series from the Texas grid and use Matpower  to compute the optimal power flow solutions to obtain the voltage phasors for the IEEE 118-bus system. All the tests  $T=10$ hours as the historical time window, i.e. 10 observed data points, to forecast voltage phasors in the next  hours $H= 1, 2, 3, 4 , 5$ where $H=0$ is a PSSE problem that estimates the complete voltage phasors $\bm{x}_{t}$ given $\bm{z}_{t-T+1}, \ldots, \bm{z}_{t}$. We compare the proposed GCN and GRN with other NNs.  In particular, the benchmark algorithm  FNN  has 4 layers with 512 neurons each  layer.  Another benchmark algorithm CNN  has three hidden layers with the 32, 64 and 32 output channels and one fully connected NNs, respectively. The third benchmark GNN \cite{kipf2017semi, hossain2021state}   is a 1st-order approximation Chebyshev GCN \cite{defferrard2016convolutional} with the adjacency matrix as GSO. The fourth benchmark RNN utilizes the RNN as the feature extraction layer, and then FNNs as the hidden and output layers. 

\begin{table}[!htb]\scriptsize
\center
\vspace{-0.4cm}
\begin{threeparttable}
\caption{PMU Installed Buses in the Transmission Network}
\label{pmu_selec1}
\centering
\begin{tabular}{c |c  }
\toprule
Systems&  Bus Name   \\
% \midrule
% 13-bus feeder  &   sourcebus.1, 680.2, 675.1, 692.2, 675.3, 646.3, \\
% with PMUs & sourcebus.3, 646.2 \\
\midrule
% 123-bus DG  & 1.1, 1.2, 1.3,  100.1, 100.2, 100.3, 101.1, 101.2, 101.3, 102.3, \\
% with PMUs& 10.1,  103.3, 104.3, 108.1, 108.2,  105.1, 105.2, 105.3, 106.2, 107.2 \\
% \midrule
118-bus   & 14,	117, 72, 86,	43,	67,	99,	87,	16,	33,	112,	28,	98,	111,	53,	97, 1	 \\
with PMUs&42, 107,	48,	22,	46,	13,	24,	101,	44,	73,	109,	29,	20,	91, 26,	84,	10\\
&	52,	57,	76,	115,	39,	74,	104,	93,  79,	35,	6,	18,	88,	60,	116,	55,	58	\\
& 68,	64,	7,	50,	103,	75,	78,	83,	69.\\
\bottomrule
\end{tabular}
\end{threeparttable}
\vspace{-0.2cm}
\end{table}
With the predicted voltage phasors, we further utilize the power flow  solver, i.e. Matpower, to obtain the feasible power generations, and then calculate their corresponding fuel costs.
The evaluation metrics for comparison includes mean square error (MSE) between the predicted and ground-truth voltage phasors and Mean Absolute Percentage Error (MAPE) between the predicted and optimal fuel costs. As shown in Table \ref{pmu_selec1}, we choose the number of sensor placements $\abs{\mathcal{M}} = 60$ and place them so as to maximize $\max_{ {\mathcal{F}}_{\mathcal{M}}}\varpi_{\min} ({\mathcal{F}}_{\mathcal{M}} \mathbf{U}_{\mathcal{K}})$. Besides, through numerous simulations for the hyperparameter tuning, we choose $\mu_1$ = 1$e$-6 and $\mu_2$ = 1$e$-3 for all benchmarks.

\begin{table}[!htb]
\scriptsize
 \renewcommand\tabcolsep{2.0pt}
\center
\begin{threeparttable}
\caption{MSE: PSSF for Voltage Phasors in Transmission Networks}
\label{forecasting1}
\centering
\begin{tabular}{c |c  c  c  c  c  c}
\toprule
Future  (Hours)&  $H=0$ &  $H=1$ &  $H=2$ & $H=3$ &  $H=4$ &  $H=5$  \\
\midrule
FNN&        1.1217$e$-4   & 4.7697$e$-4  &  7.7697$e$-4 &   5.4458$e$-4 &   9.2263$e$-4  &  8.8066$e$-3  \\
CNN&        2.8069$e$-4  & 4.4070$e$-4  & 1.7169$e$-3  &    1.7238$e$-2 &   1.6815$e$-2  &  1.6758$e$-2   \\
RNN&        8.8034$e$-4  & 8.7639$e$-4  &   7.6329$e$-4 &   7.7802$e$-4 &   7.1659$e$-4  &  7.8581$e$-4  \\
1$^{st}$GNN\cite{kipf2017semi}&    7.8849$e$-4  & 7.2899$e$-4  &    7.5874$e$-4 &   8.6065$e$-4 &   7.5357$e$-4  &  8.1019$e$-4  \\
GCN &   \textbf{6.1381\textit{e}-5}  & \textbf{1.0080\textit{e}-4}  &   2.6714$e$-4 &   2.0157$e$-4 &   3.2469$e$-4  &  2.2308$e$-4  \\
GRN &   7.2153$e$-5  & 1.8058$e$-4  &   \textbf{2.4738\textit{e}-4} &   \textbf{1.3476\textit{e}-4} &   \textbf{2.9829\textit{e}-4}  &  \textbf{2.1137\textit{e}-4}  \\
\bottomrule
\end{tabular}
\end{threeparttable}
\end{table}
% %%%
% \begin{table}[!htb]
% \vspace{-0.3cm}
% \scriptsize
%  \renewcommand\tabcolsep{2.0pt}
% \center
% \begin{threeparttable}
% \caption{Power System State Forecasting for Distribution Networks}
% \vspace{-0.1cm}
% \label{forecasting}
% \centering
% \begin{tabular}{c |c  c  c  c  c  c}
% \toprule
% Future  (Hours)&  $H=0$ &  $H=1$ &  $H=2$ & $H=3$ &  $H=4$ &  $H=5$  \\
% \midrule
% FNN&    	0.00004054&	0.0005374&	0.0006752&	0.0007041&	0.0006582&	0.0004159  \\
% CNN&    	0.0002031&	0.0005491&	0.0006357&	0.0007497&	0.0007237&	0.0007513 \\
% RNN&    	0.00008582 &	0.0006239  & 	0.0006795 &  0.0006298 &	 0.0003373  &	0.0002409 \\
% 1$^{st}$GNN\cite{kipf2017semi}&    	0.0004827&	0.0005644&	0.0006517&	0.0006167&	0.0004993&	0.0005423 \\
% GCN &	   \textbf{0.00001665}&	0.0002154&	0.0003070&	0.0003730&	0.0002084&	0.0001565 \\
% GRN &	    0.00002228&	\textbf{0.0001896}&	\textbf{0.0002693}&	\textbf{0.0003234}&	\textbf{0.0001151}&	\textbf{0.0001402} \\
% \bottomrule
% \end{tabular}
% \end{threeparttable}
% \vspace{-0.2cm}
% \end{table}
%
%

\begin{table}[!htb]
\scriptsize
 \renewcommand\tabcolsep{2.0pt}
\center
\begin{threeparttable}
\caption{MAPE: PSSF for Fuel Costs in Transmission Networks}
\vspace{-0.3cm}
\label{forecasting2}
\centering
\begin{tabular}{c |c  c  c  c  c  c}
\toprule
Future  (Hours)&  $H=0$ &  $H=1$ &  $H=2$ & $H=3$ &  $H=4$ &  $H=5$  \\
\midrule
FNN&    	2.0602\%  & 2.1338\%  & 6.5671\% &	4.5503\% &	6.9320\%  &	7.0781\%  \\
CNN&    	2.0683\%  & 9.1396\%  &	9.8635\% &	16.4026\% &	28.7120\%  &	58.2591\%  \\
RNN&    	3.0049\%  & 1.9634\% &	2.3701\% &	2.7542\% &	2.1048\%  &	2.4831\%  \\
1$^{st}$GNN\cite{kipf2017semi}&    3.7144\%  & 2.8667\%  &	2.2925\% &	2.3752\% &	2.3347\%  &	4.7779\%  \\
GCN &	\textbf{0.3838\%}   & \textbf{1.0809\%}  & \textbf{0.9991\%} &	\textbf{1.2734\%} &	2.0542\% &	2.5816\%  \\
GRN &	0.7065\%  &  0.7533\%  & 1.6357\% &	1.6365\% &	\textbf{1.5054\%}  & \textbf{1.8188\%}  \\
\bottomrule
\end{tabular}
\end{threeparttable}
\vspace{-0.5cm}
\end{table}

\subsubsection{PSSE and PSSF Results}
Tables \ref{forecasting1} and \ref{forecasting2} show the results of  GCN and GRN and various baselines on the IEEE 118-bus system   experiments described above. The results illustrate that both GCN and GRN achieve  the best performance. 
In particular,  $H = 0$ is the PSSE problem, and the MSE of \eqref{recover_vp} for estimation   is   2.3708$e$-4. While this is a respectable outcome, the supervised  GCN and GRN have much smaller error, i.e. 6.1381$e$-5 and 7.2153$e$-5, respectively. Another observation is that   the voltage phasors predicted by GCN and GRN could approximate the OPF results with much smaller MAPE, e.g. 0.9991\% and 1.6357\% at $H=2$, compared with other methods, e.g. 6.5671\% of FNN and  9.8635\% of CNN. We illustrate  examples by the GCN method for the IEEE 118-bus system in Fig. \ref{estimation_forecast} to show the ground-truth fully-observed voltage  phases with the predicted  ones with $H=0, 1$, which shows the predicted  voltage phases are very close to the ground-truth. We also  observe that the performance of GRN and GCN are similar for short-time forecasting, while GRN outperforms GCN in the long-time forecasting task (e.g. 1.5054\% compared with 2.0542\% of MAPE for $H=4$). Here, we emphasize that  the reason why GCN and GRN have  very small MSE and MAPE for forecasting is that  voltage phasors have a  constrained (low variance) distribution, i.e., power flow constraints, which has the support of the graph. Therefore, the GCN is most effective at internalizing the distribution and  approximating the Bayesian MSE estimator by capturing the spatiotemporal correlations.

\begin{table}[!htb]\scriptsize
\center
\vspace{-0.3cm}
\begin{threeparttable}
\caption{PMU Installed Buses in the Distribution Network}
\label{pmu_selec}
\centering
\begin{tabular}{c |c  }
\toprule
Systems&  Bus Name   \\
% \midrule
% 13-bus feeder  &   sourcebus.1, 680.2, 675.1, 692.2, 675.3, 646.3, \\
% with PMUs & sourcebus.3, 646.2 \\
\midrule
123-bus   & 1.1, 1.2, 1.3,  2.2, 3.3, 7.1, 7.2, 7.3, 4.3, 5.3, 6.3, 8.1, 8.2, 8.3, 10.1, 12.2 \\
DG with &   13.1, 13.2, 13.3,  9r.1, 14.1, 34.3, 18.1, 18.2, 18.3, 11.1, 15.3, 16.3, 17.3 \\
PMUs & 9.1, 19.1, 150.1, 150.2, 150.3, 150r.1, 150r.2, 150r.3,  149.1, 149.2, 149.3\\
% \midrule
% 118-bus   & 14,	117, 72, 86,	43,	67,	99,	87,	16,	33,	112,	28,	98,	111,	53,	97, 1	 \\
% with PMUs&42, 107,	48,	22,	46,	13,	24,	101,	44,	73,	109,	29,	20,	91, 26,	84,	10\\
% &	52,	57,	76,	115,	39,	74,	104,	93,  79,	35,	6,	18,	88,	60,	116,	55,	58	\\
% & 68,	64,	7,	50,	103,	75,	78,	83,	69.\\
\bottomrule
\end{tabular}
\end{threeparttable}
\vspace{-0.4cm}
\end{table}

\subsection{GCN-DRL and GRN-DRL for Voltage Control}
In this section, we compare the performance of voltage control DRL strategies using GCN and GRN that we propose with benchmark algorithms and study its learning stability in the training phase. In Fig. \ref{Figure4b} we also  validate  the efficacy of the reduced GSO \textbf{Lemma 2} in
 Section \ref{sec:reducedGSO}.

%The benchmark algorithm FNN-DRL  has 3 layers with 512 neurons each  layers. Another benchmark algorithm CNN-DRL has three  hidden layers with the 32, 64 and 32 output channels, respectively. 
\subsubsection{Policy Training}
To validate the advantages of the proposed GCN and GRN over the state of the art, the DRL scheme we showcase is an instance of the popular Proximal Policy Optimization (PPO) \cite{schulman2017proximal}.
We compare the proposed GCN-DRL and GRN-DRL architecture with existing DRL methods for voltage regulation. As PPO outputs are discrete actions, we discretize the actions space $[-1, 1]$ with  spacing $0.2$.

\subsubsection{Experiment Setup}
The DRL experiments
are run on the 123-bus feeder distribution network test case. We use demand data from Austin in the OpenEI\footnote{\url{https://data.openei.org/data_lakes#Data-Lakes-Datasets}}, and
%in the first application,  we set $K =5$ and $T= 10$ for the GCN and GRN. The  GCN and GRN are utilized to predict the voltage phasors in the future 1, 2, 3, 4, and 5 hours ($H= 1, 2, 3, 4 , 5$). In addition, $H=0$ is a distribution state estimation problem. 
historical PV data for training and testing, with three PV smart inverters installed in the load buses (Buses 51, 53, 60) in Fig. \ref{Figure4a} and six PV smart inverters  (Buses 69,  51,  52,   82,  68,  94) in Fig. \ref{Figure4c}, respectively.
We use OpenDSS to estimate the grid state.
We set the desired voltage magnitude $\bar{v} = 1$ p.u.  We test Lemma \ref{lemma1} in Section \ref{sec:reducedGSO} and apply the proposed Kron-reduction to design the GSO (see Figure \ref{Figure4b}).  The measurement buses, shown in TABLE \ref{pmu_selec}, are selected according to the algorithm in Section III.A.  In particular, we select  40 phases from 278 phases ($\approx$ 14\% of the buses) in the three-phase 123-bus feeder system.  Considering that real distribution feeders include thousands of buses, this would bring the cost for PMU measurements systems to reasonable levels. 
The DRL parameters are as follows.  The learning rate is 0.0007. The discounted factor $\gamma$ is 0.99. The PPO clip parameter $\epsilon$ is 0.1, the entropy loss weight is 0.01 and value loss weight is 1.  There are 10 spatial and temporal channels for both GCN and GRN layers performing the feature extraction, followed by 512 neurons in an  FNN layer followed by the output layer.

\subsubsection{DRL regulation results}
The learning curves of GRN-DRL, GCN-DRL, FNN-DRL, and CNN-DRL with the full  and partial observations are shown in  Fig. \ref{Figure4a} and Fig. \ref{Figure4b}, respectively. The two figures show the average training reward, where the bands represent the standard deviation over 5 runs.
In particular, the results in Fig. \ref{Figure4a} show that the voltage  deviations  of GRN and GCN, i.e., $\sum_{n_\phi \in \mathcal{N}_s} \abs{r_{n_\phi}}$,  are    0.0332 p.u. and  0.0398 p.u., which outperform FNN and CNN that  have 0.0504 p.u. and 0.0657 p.u., respectively. Another observation is that GRN and GCN are competitive in convergence time and performance. 
With the partial observations (40 out of 278), the results in Fig.\ref{Figure4b}  show that  $\sum_{n_\phi \in \mathcal{N}_s} \abs{r_{n_\phi}}$ of   GCN and GRN converge into 0.0492 p.u.  and 0.0469 p.u., respectively. However, $\sum_{n_\phi \in \mathcal{N}_s} \abs{r_{n_\phi}}$ of FNN and CNN converge into 0.0784 p.u.  and  0.0697 p.u., respectively.  
%When we scale these policies 
With 6 smart inverters, the test in Fig. \ref{Figure4c} shows the learning curves of  FNN and CNN decreases after 50 episodes, which indicates that they trigger the deadly triad of DRL. In contrast, GCN and GRN continue to have excellent performance, demonstrating that they do enhance the stability of DRL.

\vspace{-0.20cm}

\section{Conclusions}
In this paper we proposed novel physics-aware   GCN and GRN  frameworks for single and three-phase power systems. The proposed architectures are shown to be more effective than conventional NNs in extracting spatio-temporal features from the voltage phasors, in  forecasting and control applications.  Moreover, we show that even having roughly 14\% of the state values measurements leads to excellent performance compared to other benchmarks, i.e., FNN, CNN, RNN and 1st-GCN for the aforementioned applications.

\appendix

\subsection{Proof of Lemma 1}
Here we obtain two decoupled real equations  describing the dependence between the active and reactive power and the magnitude and phases of the state vector. %$\bm{v}_n\bm{v}_{n}^H$ and $\bm{v}_n\bm{v}_{m}^H$. 
To do so, we will be using the expansion $e^{jx} = 1+jx$ for phase terms of the three-phase state sub-vectors in the products $\bm{v}_n\bm{v}_{n}^H$ and $\bm{v}_n\bm{v}_{m}^H$, after re-centering them around the phases of a balanced system. Let $\Psi_n^{(3)} \triangleq \text{diag}([1, e^{-\mathfrak{j}2\pi/3}, e^{\mathfrak{j}2\pi/3}]^\top)$, 
%where   $\text{diag}(\cdot)$ which is a diagonal matrix with diagonal elements equal to the entries of the vector in the argument of $\text{diag}(\cdot)$.  We denote $\circ$ as the Hadamard product. Let
$c_c=\cos(\frac{2\pi}{3})$ and $c_s=\sin(\frac{2\pi}{3})$. $\mathbb{1}$ is the all-ones vector and $\mathbb{1}\mathbb{1}^\top$ is the all-ones matrix.
In the following, we assume $\Psi_n^{(3)} = \Psi_m^{(3)}$ but we could similarly account for other shifts modeling  specific electrical elements, such as transformers.
We will make use of the following propositions. With $\mathbf{A}$, $\mathbf{B}$, $\mathbf{C}$ and $\mathbf{E}$ real square matrices, and $\bm{a}$ and $\bm{b}$ are real vectors, the following holds: 
\begin{proposition}
	(\textbf{P2}) If $\mathbf{C}$ and $\mathbf{E}$ are diagonal matrices, then $\mathbf{C}(\mathbf{A}\circ \mathbf{B})\mathbf{E} = \mathbf{A}\circ (\mathbf{C}\mathbf{B}\mathbf{E})$.
\end{proposition}
\begin{corollary}
	(\textbf{C1}) If $\mathbf{C}$ and $\mathbf{E}$ are diagonal matrices, then $\mathbf{C}\mathbf{A}\mathbf{E} = \mathbf{A}\circ (\mathbf{C}(\mathbb{1}\mathbb{1}^\top)\mathbf{E}) =\mathbf{A}\circ (\text{diag}(\mathbf{C})(\text{diag}(\mathbf{E}))^\top ) $.
\end{corollary}
\begin{proposition}
	(\textbf{P4}) $D(\mathbf{A}\mathbf{B})=  \sum_j (\mathbf{A}\circ \mathbf{B}^\top)_{ij}$.
\end{proposition}
\begin{proposition}
	(\textbf{P5}) If $\mathbf{B}$ and $\mathbf{C}$ are symmetric, $D((\mathbf{A}\circ \mathbf{B})\mathbf{C}) = D(\mathbf{A}(\mathbf{B}\circ \mathbf{C})) =  D(\mathbf{A}(\mathbf{B}\circ \mathbf{C})^\top)$.
\end{proposition}
\emph{proof:}
\begin{align}
&D((\mathbf{A}\circ \mathbf{B})\mathbf{C}) = \sum_i((\mathbf{A}\circ \mathbf{B}) \circ \mathbf{C}^\top)_{ij} =  \sum_i(\mathbf{A}\circ (\mathbf{B} \circ \mathbf{C}^\top))_{ij}\notag\\
&= \sum_i(\mathbf{A}\circ (\mathbf{B}^\top \circ \mathbf{C}^\top))_{ij}= \sum_i(\mathbf{A}\circ ( \mathbf{C}^\top \circ \mathbf{B}^\top ) )_{ij}     \notag\\
&=\sum_i(\mathbf{A}\circ ( \mathbf{C} \circ \mathbf{B} )^\top )_{ij} = D(\mathbf{A} ( \mathbf{C} \circ \mathbf{B} ))= D(\mathbf{A} ( \mathbf{B} \circ  \mathbf{C}  )).\notag
\end{align}
\begin{proposition}
	(\textbf{P5}) $D(\bm{a}\bm{b}^T)\!=\! \text{diag}(\bm{b})\bm{a}$ ,  $D(\mathbf{A})\!=\!D(\mathbf{A}^\top)$.
\end{proposition}

Now, we are ready to introduce how to design the GSO. We will refer to
the specific propositions or corollary in each equation, such as {\color{blue}\textbf{P2}} or {\color{blue}\textbf{C1}}, with blue color.
By adding and subtracting from the phase angle in $\bm{v}_n$, we obtain:
\begin{equation}
	\bm{v}_n = \Psi_n^{(3)} \text{diag}(\abs{\bm{v}_{n}}) \begin{bmatrix} e^{\mathfrak{j}\varphi_{n_a}}  \\  e^{\mathfrak{j}\varphi_{n_b}} \\  e^{\mathfrak{j}\varphi_{n_c}}   \end{bmatrix}\\
	= \Psi_n^{(3)}  \text{diag}(\abs{\bm{v}_n}) e^{\mathfrak{j}\bm{\varphi}_n}
\end{equation}
Therefore, the outer product $\bm{v}_n\bm{v}_m^H$ are: 
\begin{align}
	&\bm{v}_n\bm{v}_m^H =  \Psi_n^{(3)}  \text{diag}(\abs{\bm{v}_n}) e^{\mathfrak{j}(\bm{\varphi}_n\mathbb{1}^\top - \mathbb{1}\bm{\varphi} _m^\top )} \text{diag}(\abs{\bm{v}_m}) (\Psi_m^{(3)} )^H \notag\\
	& = \overbrace{\text{diag}(\abs{\bm{v}_n}) \Psi_n^{(3)}}^{{\color{blue}\mathbf{C} \text{ in }  \textbf{C1}}} 
	\big( \overbrace{\mathbb{1}\mathbb{1}^\top + \mathfrak{j}(\bm{\varphi}_n\mathbb{1}^\top - \mathbb{1}\bm{\varphi} _m^\top )}^{{\color{blue}\mathbf{A} \text{ in }  \textbf{C1}}}\big)   \label{outv1} \\
	& \big[\overbrace{(\Psi_m^{(3)} )^H \text{diag}(\abs{\bm{v}_m})}^{{\color{blue}\mathbf{E} \text{ in }  \textbf{C1}}}\big] \stackrel{\color{blue}\textbf{C1}}{=} \left(\mathbb{1}\mathbb{1}^\top + \mathfrak{j}(\bm{\varphi}_n\mathbb{1}^\top \!\!\!-\! \mathbb{1}\bm{\varphi} _m^\top )\right) \\
	&\circ\Big(\text{diag}(\abs{\bm{v}_n})
	\Gamma \text{diag}(\abs{\bm{v}_m})\Big)\label{outv2}
\end{align}
where $\Gamma\triangleq \Psi_n^{(3)} (\mathbb{1}\mathbb{1}^\top) (\Psi_m^{(3)} )^H$ and can be expressed as:
\begin{equation}
\begin{split}\small
 [\Gamma]_{k n} = e^{\mathfrak{j}\frac{2(k - n)\pi}{3}} = [\Gamma_c]_{kn}  + \mathfrak{j}[\Gamma_s]_{kn},~ k,n \in\{0,1,2\}\notag
\end{split}
\end{equation}
where it is easy to verify that $ \Gamma_c=\Gamma_c^\top$.
%$ \Gamma_c=  \begin{bmatrix} 
% 1 & c_c & c_c\\ 
%c_c & 1 &  c_c\\
%c_c & c_c &  1
% \end{bmatrix} $ and $\Gamma_s = %\begin{bmatrix} 
% 0 & c_s & -c_s\\ 
%-c_s & 0 &  c_s\\
%c_s & -c_s &  0
% \end{bmatrix} $.
 
%  \mathfrak{j}  \overbrace{\begin{bmatrix} 
%  1 & c_s & -c_s\\ 
% -c_s & 1 &  c_s\\
% c_s & -c_s &  1
%  \end{bmatrix}}^{\Gamma_s}

% Our derivations rely on the following facts:
\subsection{Proof of Proposition 1}

\subsubsection{Active Power GSO}
Next we use the approximation in developing the component relative to the phase term we use the approximation\footnote{This is effective but not truly necessary since the multiplication with $\text{diag}(\abs{\bm{v}_n})$ could be used as part of the definition of the graph signal. } that  $\abs{\bm{v}_n} \approx \mathbb{1}$ and $\abs{\bm{v}_m} \approx \mathbb{1}$ in \eqref{outv2}. With this approximation, we substitute \eqref{outv2} in \eqref{sinj}. Therefore, the real part of the first term in \eqref{sinj} is
\begin{align}
	&\Re\bigg\{-D\bigg(\bm{v}_n\bm{v}_n^H \big(\frac{\mathfrak{j}}{2} \mathbf{B}^{s}_{mn} + \mathfrak{j}\mathbf{B}^{(n)}_{mn} \big)\bigg)\bigg\} = - D\bigg( \Re\bigg\{\Big(\big[\mathfrak{j}\mathbb{1}\mathbb{1}^\top   \notag\\
	&   (\bm{\varphi}_n\mathbb{1}^\top - \mathbb{1}\bm{\varphi}_n^\top )\big] \circ \Gamma \Big) \big(\frac{1}{2} \mathbf{B}^{s}_{mn} + \mathbf{B}^{(n)}_{mn} \big)  \bigg\} \bigg)=D\bigg(\Big(\mathbb{1}\mathbb{1}^\top  \notag\\
	&    \circ \Gamma_s + (\bm{\varphi}_n\mathbb{1}^\top - \mathbb{1}\bm{\varphi}_n^\top )\circ \Gamma_c  \Big)  \big(\frac{1}{2} \mathbf{B}^{s}_{mn} + \mathbf{B}^{(n)}_{mn} \big) \bigg)\label{pinj1st}
%	&  \left((\mathbb{1}\mathbb{1}^\top) + \mathfrak{j}(\bm{\varphi}_n\mathbb{1}^\top - \mathbb{1}\bm{\varphi} _m^\top )\right)  \circ (\Psi_n^{(3)} (\mathbb{1}\mathbb{1}^\top) (\Psi_m^{(3)} )^H ) \left(\frac{\mathfrak{j}}{2} \mathbf{B}^{s}_{mn} + \mathfrak{j}\mathbf{B}_{mn} \right)
\end{align}
We separate the biased part that does not involve in $(\bm{\varphi}_n\mathbb{1}^\top - \mathbb{1}\bm{\varphi}_n^\top )$ from \eqref{pinj1st}, and define it as:
\begin{align}
	&\bm{p}^{inc}_n\triangleq D\bigg(\Gamma_s   \big(\frac{1}{2} \mathbf{B}^{s}_{mn} + \mathbf{B}^{(n)}_{mn} \big) \bigg)\label{pinc}.
%	&  \left((\mathbb{1}\mathbb{1}^\top) + \mathfrak{j}(\bm{\varphi}_n\mathbb{1}^\top - \mathbb{1}\bm{\varphi} _m^\top )\right)  \circ (\Psi_n^{(3)} (\mathbb{1}\mathbb{1}^\top) (\Psi_m^{(3)} )^H ) \left(\frac{\mathfrak{j}}{2} \mathbf{B}^{s}_{mn} + \mathfrak{j}\mathbf{B}_{mn} \right)
\end{align}
where $ \mathbb{1}\mathbb{1}^\top \circ \Gamma_s = \Gamma_s$.
The remaining part of \eqref{pinj1st} involving $(\bm{\varphi}_n\mathbb{1}^\top - \mathbb{1}\bm{\varphi}_n^\top )$, denoted by $\bm{p}^{in}_n$,   can be expressed as:
\begin{align}
	&\bm{p}^{in}_n \triangleq D\bigg(\Big( \overbrace{(\bm{\varphi}_n\mathbb{1}^\top \!-\! \mathbb{1}\bm{\varphi}_n^\top )}^{{\color{blue}\mathbf{A} \text{ in }  \textbf{P4}}} \circ \overbrace{\Gamma_c}^{{\color{blue}\mathbf{B} \text{ in }  \textbf{P4}}}  \Big)  \overbrace{\big(\frac{1}{2} \mathbf{B}^{s}_{mn} + \mathbf{B}^{(n)}_{mn} \big)}^{{\color{blue}\mathbf{B} \text{ in }  \textbf{P4}}} \bigg) \label{eq45pin} \\
% 	& =\sum_j\left(\left( (\bm{\varphi}_n\mathbb{1}^\top - \mathbb{1}\bm{\varphi}_n^\top )\circ \Gamma_c  \right)  \circ \left(\frac{1}{2} \mathbf{B}^{s}_{mn} + \mathbf{B}^{(n)}_{mn} \right)^\top \right)_{ij}\label{eq46pin}\\
% 	& = \sum_j\left(\left( \bm{\varphi}_n\mathbb{1}^\top - \mathbb{1}\bm{\varphi}_n^\top \right)\circ  \left(\frac{1}{2} \hat{\mathbf{B}}^{s}_{mn} + \hat{\mathbf{B}}^{(n)}_{mn} \right)\right)_{ij} \label{eq47pin}\\
	& \stackrel{\color{blue}\textbf{P4}}{=}  D\bigg(\big( \bm{\varphi}_n\mathbb{1}^\top - \mathbb{1}\bm{\varphi}_n^\top \big) \big(\frac{1}{2} \!\!\overbrace{\hat{\mathbf{B}}^{s}_{mn}}^{\triangleq \Gamma_c \circ {\mathbf{B}}^{s}_{mn}} + \overbrace{\hat{\mathbf{B}}^{(n)}_{mn}}^{\triangleq \Gamma_c \circ {\mathbf{B}}^{(n)}_{mn}}\big)^\top\bigg) \label{eq48pin}\\
%	& =	D\Bigg\{ \bm{\varphi}_n   \Bigg[\Bigg(\frac{1}{2} \hat{\mathbf{B}}^{s}_{mn} + \hat{\mathbf{B}}^{(n)}_{mn}\Bigg) \mathbb{1}\Bigg]^\top\Bigg\}  \notag\\
%	& -	D\Bigg\{ \mathbb{1}  \Bigg[\Bigg(\frac{1}{2} \hat{\mathbf{B}}^{s}_{mn} + \hat{\mathbf{B}}^{(n)}_{mn}\Bigg) \bm{\varphi}_n \Bigg]^\top\Bigg\}  \notag\\
& \stackrel{\color{blue}\textbf{P5}}{=}  \text{diag}\bigg( \big(\frac{1}{2} \hat{\mathbf{B}}^{s}_{mn} + \hat{\mathbf{B}}^{(n)}_{mn}\big) \mathbb{1} \bigg) \bm{\varphi}_n - \big(\frac{1}{2} \hat{\mathbf{B}}^{s}_{mn} + \hat{\mathbf{B}}^{(n)}_{mn}\big)  \bm{\varphi}_n. \label{eq49pin}
\end{align}
% From \eqref{eq45pin} to \eqref{eq46pin} and from \eqref{eq47pin} to \eqref{eq48pin}, we utilize this proposition, $D(AB)=  \sum_j (A\circ B^\top)_{ij}$. From \eqref{eq46pin} to \eqref{eq47pin}, 
As the Hadamard product commutes and $\Gamma_c$, $\mathbf{B}^{s}_{mn}$ and $\mathbf{B}_{mn}^{(n)}$  are symmetric,  $\hat{\mathbf{B}}^{s}_{mn}$ and $\hat{\mathbf{B}}^{(n)}_{mn}$ are  symmetric. 
% From \eqref{eq48pin} to \eqref{eq49pin}, we use this property: $D(\bm{a}\bm{b}^T) = \text{diag}(\bm{b})\bm{a}$ and $D(A) = D(A^\top)$.
%  {\color{red}Here, we need to emphasize that the matrices $\bm{\varphi}_n\mathbb{1}^\top$ and $\mathbb{1}\bm{\varphi}_n^\top$ are rank-1 so that they could be expressed as the outer product of two vectors. If we incorporate the $\Gamma_c$ into the matrices $\bm{\varphi}_n\mathbb{1}^\top$ and $\mathbb{1}\bm{\varphi}_n^\top$, i.e., $\bm{\varphi}_n\mathbb{1}^\top \circ \Gamma_c$ and  $\mathbb{1}\bm{\varphi}_n^\top \circ \Gamma_c$,  both $\bm{\varphi}_n\mathbb{1}^\top \circ \Gamma_c$ and  $\mathbb{1}\bm{\varphi}_n^\top \circ \Gamma_c$ are no longer rank-1. As a result, $\bm{\varphi}_n\mathbb{1}^\top \circ \Gamma_c$ and  $\mathbb{1}\bm{\varphi}_n^\top \circ \Gamma_c$ cannot be regarded as  graph signals because graph signals in \eqref{gs} are a non-zero vector that is always rank-1. Therefore, we modify the GSO instead of graph signals.}
Replacing $m$ with $n$ and $\big(\frac{1}{2} \hat{\mathbf{B}}^{s}_{mn} + \hat{\mathbf{B}}^{(n)}_{mn}\big)$ with $ \hat{\mathbf{B}}^{(m)}_{mn}$, this result applies also to $\bm{v}_{n}\bm{v}_m^H$. Thus, the real part of the $2^{nd}$ term in \eqref{sinj} is
\begin{align}
&\Re\left\{-D\Big(\bm{v}_n\bm{v}_m^H ~ \mathfrak{j}\mathbf{B}^{(m)}_{mn} \Big) \right\} = \notag\\
&=D\bigg(\Big( \Gamma_s + (\bm{\varphi}_n\mathbb{1}^\top - \mathbb{1}\bm{\varphi}_m^\top )\circ \Gamma_c  \Big)   \mathbf{B}^{(m)}_{mn}  \bigg)\label{pinj2st}
\end{align}
Likewise, we separate the bias, that is not part of $(\bm{\varphi}_n\mathbb{1}^\top - \mathbb{1}\bm{\varphi}_m^\top )$ from \eqref{pinj2st}, and define it as:
\begin{align}
	&\bm{p}^{ouc}_n\triangleq D\left(\Gamma_s    \mathbf{B}^{(m)}_{mn}  \right)\label{pinc}.
%	&  \left((\mathbb{1}\mathbb{1}^\top) + \mathfrak{j}(\bm{\varphi}_n\mathbb{1}^\top - \mathbb{1}\bm{\varphi} _m^\top )\right)  \circ (\Psi_n^{(3)} (\mathbb{1}\mathbb{1}^\top) (\Psi_m^{(3)} )^H ) \left(\frac{\mathfrak{j}}{2} \mathbf{B}^{s}_{mn} + \mathfrak{j}\mathbf{B}_{mn} \right)
\end{align}
The remaining part of \eqref{pinj2st} that involves in $(\bm{\varphi}_n\mathbb{1}^\top - \mathbb{1}\bm{\varphi}_m^\top )$, denoted by $\bm{p}^{in}_n$,   can be expressed as:
\begin{align}
&D\bigg(\Big( (\bm{\varphi}_n\mathbb{1}^\top - \mathbb{1}\bm{\varphi}_m^\top )\circ \Gamma_c  \Big)   \mathbf{B}^{(n)}_{mn}  \bigg)\notag\\
& = \text{diag}\left(  \hat{\mathbf{B}}^{(m)}_{mn} \mathbb{1}\right) \bm{\varphi}_n -  \hat{\mathbf{B}}^{(m)}_{mn}  \bm{\varphi}_m 
\end{align}
Finally, by excluding $\bm{p}_n^{inc}$ and $\bm{p}_n^{ouc}$ from GSO, we have
\begin{align}
&\tilde{\bm{p}}_n \triangleq  \bm{p}_n - \bm{p}_n^{inc} - \bm{p}_n^{ouc} =  \bm{p}_n - \bm{p}_n^{cst} 	\notag\\
&\sum_{m\in \mathcal{N}_n} \bigg(\Big(\text{diag}\big( \frac{1}{2} \hat{\mathbf{B}}^{s}_{mn} + \hat{\mathbf{B}}^{(n)}_{mn}\big)\mathbb{1} +   \hat{\mathbf{B}}^{(m)}_{mn}\mathbb{1}\Big)  \bm{\varphi}_n \notag\\
&- \big(\frac{1}{2} \hat{\mathbf{B}}^{s}_{mn} + \hat{\mathbf{B}}^{(n)}_{mn}\big)  \bm{\varphi}_n - \hat{\mathbf{B}}^{(m)}_{mn}  \bm{\varphi}_m  \bigg)
\end{align}
To write it in a compact way, we have
\begin{align}
\tilde{\bm{p}} = \mathbf{\hat{B}} \bm{\varphi},
\end{align}
where $\tilde{\bm{p}}$, $\hat{\mathbf{B}}$ an $\bm{\varphi}$ are denoted by
\begin{align}
&\tilde{\bm{p}} =  \begin{bmatrix}
\tilde{\bm{p}}_1^\top,\cdots,
\tilde{\bm{p}}_{\abs{\mathcal{N}}}^\top
\end{bmatrix}^\top,
 \hat{\mathbf{B}} \triangleq ((\mathbb{1}\mathbb{1}^\top)_{N}  \otimes \Gamma_c)\circ \mathbf{B}\\
 &\bm{\varphi} =  \begin{bmatrix}
\bm{\varphi}_1^\top,\cdots,\bm{\varphi}_{\abs{\mathcal{N}}}^\top
\end{bmatrix}^\top,
\end{align}
and $\mathbf{\hat{B}}$ has the same structure with $\mathbf{{B}}$ with replacing ${\mathbf{B}}^s_{mn}$, ${\mathbf{B}}^{(n)}_{mn}$ and ${\mathbf{B}}^{(m)}_{mn}$ with $\hat{\mathbf{B}}^s_{mn}$, $\hat{\mathbf{B}}^{(n)}_{mn}$ and $\hat{\mathbf{B}}^{(m)}_{mn}$, respectively.

\subsubsection{Reactive Power GSO}
The reactive power analysis is similar to the active power analysis. In particular, we use the approximation that $\bm{\varphi}_n\mathbb{1}^\top - \mathbb{1}\bm{\varphi} _m^\top \approx \mathbf{0}$ in \eqref{outv2}, where $\mathbf{0}$ is the all-zeros matrix.  Therefore, the first  part of \eqref{sinj} is
\begin{align}
	&-D\bigg(\bm{v}_n\bm{v}_n^H \big(\frac{\mathfrak{j}}{2} \mathbf{B}^{s}_{mn} + \mathfrak{j}\mathbf{B}^{(n)}_{mn} \big)\bigg) \approx  -D \bigg(\mathbb{1}\mathbb{1}^\top \circ\notag\\
	& \Big(\overbrace{\text{diag}(\abs{\bm{v}_n})}^{{\color{blue}\mathbf{C} \text{ in }  \textbf{C1}}} \overbrace{\Gamma}^{{\color{blue}\mathbf{A} \text{ in }  \textbf{C1}}} \overbrace{\text{diag}(\abs{\bm{v}_n})}^{{\color{blue}\mathbf{E} \text{ in }  \textbf{C1}}}\Big) \big(\frac{\mathfrak{j}}{2} \mathbf{B}^{s}_{mn} + \mathfrak{j}\mathbf{B}^{(n)}_{mn} \big)\bigg)\notag\\
	&\stackrel{\color{blue}\textbf{C1}}{=} -D\bigg(\Big( (\abs{\bm{v}_n} \abs{\bm{v}_n}^\top   ) \circ (\Gamma_c + \mathfrak{j} \Gamma_s)\Big) \big(\frac{\mathfrak{j}}{2} \mathbf{B}^{s}_{mn} + \mathfrak{j}\mathbf{B}^{(n)}_{mn} \big)\bigg)\label{eq54re1}
\end{align}
Then we take the imaginary part of \eqref{eq54re1} as:
\begin{align}
	& -D\bigg(\Big( (\abs{\bm{v}_n} \abs{\bm{v}_n}^\top  )  \circ \Gamma_c  \Big)\big(\frac{1}{2} \mathbf{B}^{s}_{mn} + \mathbf{B}^{(n)}_{mn} \big)\bigg)\label{eq58q}\\
	& =-D\bigg((\abs{\bm{v}_n} \abs{\bm{v}_n}^\top  ) \circ \big(\frac{1}{2} \hat{\mathbf{B}}^{s}_{mn} + \hat{\mathbf{B}}^{(n)}_{mn} \big)\bigg)\label{eq59q}\\
	& \approx -D\bigg(  (\mathbb{1} \abs{\bm{v}_n}^\top   ) \circ \big(\frac{1}{2} \hat{\mathbf{B}}^{s}_{mn} + \hat{\mathbf{B}}^{(n)}_{mn} \big)\bigg).\label{eq59qv}	
\end{align}
 The process of transformation from \eqref{eq58q} to \eqref{eq59q} is similar to the transformation  from \eqref{eq45pin} to \eqref{eq48pin}. 
 From \eqref{eq59q} to \eqref{eq59qv}, we relax $\abs{\bm{v}_n} \abs{\bm{v}_n}^\top \approx \mathbb{1} \abs{\bm{v}_n}^\top$ by $\abs{\bm{v}_{n}}\approx \mathbb{1}$.

Replacing $m$ with $n$ this form applies also to $\bm{v}_{n}\bm{v}_m^H$. Therefore, the second imaginary  part of \eqref{sinj} is
\begin{align}
&\Im\left\{-D\left(\mathfrak{j}\bm{v}_n\bm{v}_m^H \mathbf{B}^{(m)}_{mn} \right)\right\} \approx -D\left(\big(\mathbb{1} \abs{\bm{v}_m}^\top   \big) \circ \hat{\mathbf{B}}^{(m)}_{mn}\right)\label{eq60q}
\end{align}
By summing \eqref{eq59qv} and \eqref{eq60q} together, we have 
\begin{align}
\bm{q}_{n} \approx 
-D\bigg((\mathbb{1} \abs{\bm{v}_n}^\top \pm \abs{\bm{v}_n} \mathbb{1}^\top )  \circ 
 \big(\frac{1}{2} \hat{\mathbf{B}}^{s}_{mn}\notag  \\
 + \hat{\mathbf{B}}^{(n)}_{mn} \big)\bigg) -  D\bigg((\mathbb{1} \abs{\bm{v}_m}^\top) \circ\hat{\mathbf{B}}^{(m)}_{mn} \bigg), \label{eq61}
\end{align}
where we add  and minus this item, i.e., $\abs{\bm{v}_n}\mathbb{1}^\top$, in order to split \eqref{eq61} into three parts, i.e., $\tilde{\bm{q}}_n^{ouc}$, $\tilde{\bm{q}}_n^{inc}$ and $\tilde{\bm{q}}_n^{cst}$. Specifically: 
%%%%%
%%%Part 1
%%%%%
\begin{align}
&  \tilde{\bm{q}}_n^{inc} \triangleq   -D\bigg((\mathbb{1} \abs{\bm{v}_n}^\top - \abs{\bm{v}_n}\mathbb{1}^\top  )  \circ  \big(\frac{1}{2} \hat{\mathbf{B}}^{s}_{mn} + \hat{\mathbf{B}}^{(n)}_{mn} \big)\bigg) \label{eqqinc}\\
& =D\bigg(( \abs{\bm{v}_n} \mathbb{1}^\top - \mathbb{1} \abs{\bm{v}_n}^\top  ) \big(\frac{1}{2} \hat{\mathbf{B}}^{s}_{mn} + \hat{\mathbf{B}}^{(n)}_{mn} \big)^\top\bigg) \stackrel{\color{blue}\textbf{P4}\&\textbf{P5}}{=}  \label{eqqinc2}\\  
&  \text{diag}\bigg( \big(\frac{1}{2} \hat{\mathbf{B}}^{s}_{mn} + \hat{\mathbf{B}}^{(n)}_{mn}\big)\mathbb{1}\bigg) \abs{\bm{v}_n} - \big(\frac{1}{2} \hat{\mathbf{B}}^{s}_{mn} + \hat{\mathbf{B}}^{(n)}_{mn}\big)  \abs{\bm{v}_n} \label{eqqinc3},
\end{align}
where the transformation from \eqref{eqqinc2} to \eqref{eqqinc3} is the same one from \eqref{eq48pin} to \eqref{eq49pin}.
%%%%%
Likewise,  we could replace $m$ with $n$ and $\left(\frac{1}{2} \hat{\mathbf{B}}^{s}_{mn} + \hat{\mathbf{B}}^{(n)}_{mn}\right)$ with $ \hat{\mathbf{B}}^{(m)}_{mn}$, and have the second part:
%%%%%
\begin{align}
& \tilde{\bm{q}}_n^{ouc} \triangleq D\bigg(( \abs{\bm{v}_n}\mathbb{1}^\top)  \circ     \overbrace{ \hat{\mathbf{B}}^{(m)}_{mn}}^{-\hat{\mathbf{B}}^{(n)}_{mn}} \bigg)  -D\left(( \mathbb{1}^\top \abs{\bm{v}_m}) \circ \hat{\mathbf{B}}^{(m)}_{mn} \right) \notag\\
&=\text{diag}\left( \hat{\mathbf{B}}^{(m)}_{mn} \mathbb{1}\right)  \abs{\bm{v}_n} - \hat{\mathbf{B}}^{(m)}_{mn} \abs{\bm{v}_m}\label{eqqouc}
\end{align}
%%%%%
%%%Part 3
%%%%%
% \begin{align}
% {\bm{q}}_n^{cst} \triangleq -D\bigg(\Big( \text{diag}(\abs{\bm{v}_n})(\mathbb{1}\mathbb{1}^\top)   \Big)  \circ  \left(\frac{1}{2} \hat{\mathbf{B}}^{s}_{mn}  \right)\bigg)
% \end{align}

% We relax $\text{diag}(\abs{\bm{v}_n})(\mathbb{1}\mathbb{1}^\top) \text{diag}(\abs{\bm{v}_n}) \approx (\mathbb{1}\mathbb{1}^\top) \text{diag}(\abs{\bm{v}_n})$ in Eq. \eqref{eqqinc} by $\abs{\bm{v}_{n}}\approx \mathbb{1}$  as:
% \begin{align}
% & {\bm{q}}_n^{inc} \approx D\bigg(\Big( \text{diag}(\abs{\bm{v}_n}) (\mathbb{1}\mathbb{1}^\top) - (\mathbb{1}\mathbb{1}^\top) \text{diag}(\abs{\bm{v}_n})  \Big) \bigg(\frac{1}{2} \hat{\mathbf{B}}^{s}_{mn} + \hat{\mathbf{B}}^{(n)}_{mn} \bigg)\bigg)\notag\\
% & =D\bigg(\Big( \abs{\bm{v}_n} \mathbb{1}^\top - \mathbb{1} \abs{\bm{v}_n}^\top  \Big) \bigg(\frac{1}{2} \hat{\mathbf{B}}^{s}_{mn} + \hat{\mathbf{B}}^{(n)}_{mn} \bigg)\bigg)= \label{eqqinc2}\\  
% &  \overbrace{\text{diag}\left(\mathbb{1}^\top  \left(\frac{1}{2} \hat{\mathbf{B}}^{s}_{mn} + \hat{\mathbf{B}}^{(n)}_{mn}\right)\right) \abs{\bm{v}_n} - \left(\frac{1}{2} \hat{\mathbf{B}}^{s}_{mn} + \hat{\mathbf{B}}^{(n)}_{mn}\right)  \abs{\bm{v}_n}}^{\triangleq \tilde{\bm{q}}_n^{inc}} \label{eqqinc3},
% \end{align}
% where the transformation from \eqref{eqqinc2} to \eqref{eqqinc3} is the same as the transformation from \eqref{eq48pin} to \eqref{eq49pin}.
The remaining part of Eq. \eqref{eq61} is
\begin{align}
\tilde{\bm{q}}_n^{cst} \triangleq -D\bigg((\abs{\bm{v}_n}\mathbb{1}^\top )  \circ   \frac{1}{2} \hat{\mathbf{B}}^{s}_{mn}   \bigg)
\end{align}
With $\abs{\bm{v}_{n}}\approx \mathbb{1}$, ${\bm{q}}_n^{cst}$ could be relaxed as a biased part that does not involve in $\abs{\bm{v}_n} \mathbb{1}^\top - \mathbb{1} \abs{\bm{v}_m}^\top $:
\begin{align}
\tilde{\bm{q}}_n^{cst} \approx {\bm{q}}_n^{cst} \triangleq   -D\bigg( \frac{1}{2} \hat{\mathbf{B}}^{s}_{mn}  \bigg) 
\end{align}
%%%%%%%%%
Finally, by excluding ${\bm{q}}_n^{cst}$ from the GSO, we have
\begin{align}
&\bm{q}_n - {\bm{q}}_n^{cst}  \approx \tilde{\bm{q}}_n^{inc}+ \tilde{\bm{q}}_n^{ouc} \triangleq \tilde{\bm{q}}_n =	\notag\\
&\sum_{m\in \mathcal{N}_n} \bigg(\text{diag}\Big( \big(\frac{1}{2} \hat{\mathbf{B}}^{s}_{mn} + \hat{\mathbf{B}}^{(n)}_{mn}\big) \mathbb{1} +   \hat{\mathbf{B}}^{(m)}_{mn}\mathbb{1}\Big)  \abs{\bm{v}_n} \notag\\
&- \big(\frac{1}{2} \hat{\mathbf{B}}^{s}_{mn} + \hat{\mathbf{B}}^{(n)}_{mn}\big)  \abs{\bm{v}_n} -  \hat{\mathbf{B}}^{(m)}_{mn}  \abs{\bm{v}_m}  \bigg)
\end{align}
In the same way with  active power injects, we have
\begin{align}
\tilde{\bm{q}} &= \mathbf{\hat{B}} \abs{\bm{v}},
%\end{align}
%where $\tilde{\bm{q}}$ an $\abs{\bm{v}}$ are denoted by
%\begin{align}
\\
\tilde{\bm{q}} &=  \begin{bmatrix}
\tilde{\bm{q}}_1^\top,\cdots,
\tilde{\bm{q}}_{\abs{\mathcal{N}}}^\top
\end{bmatrix}^\top,
\abs{\bm{v}} =  \begin{bmatrix}
\abs{\bm{v}_1}^\top, \cdots ,
\abs{\bm{v}_{\abs{\mathcal{N}}}}^\top
\end{bmatrix}^\top.
\end{align}

\vspace{-0.2cm}
\begin{footnotesize}

\end{footnotesize}

\end{document}